\documentclass[review,onefignum,onetabnum]{siamart171218}
\usepackage{multirow}
\usepackage{amsmath}
\usepackage{color}

\usepackage{amssymb}
\usepackage{mathrsfs}
\def\Q{\mathbf{Q}}
\def\P{\mathbf{P}}
\def\I{\mathbf{I}}

\def\zhat{\mathbf{z}}

\def\xhat{\mathbf{x}}
\def\yhat{\mathbf{y}}

\def\V{\mathbf{V}}
\def\_v{\mathbf{v}}

\def\S{\mathbf{S}}

\usepackage{lipsum}
\usepackage{amsfonts}
\usepackage{graphicx}
\usepackage{epstopdf}
\usepackage{algorithmic}
\newsiamremark{remark}{Remark}
\newsiamremark{hypothesis}{Hypothesis}
\crefname{hypothesis}{Hypothesis}{Hypotheses}
\newsiamthm{claim}{Claim}

\newcommand{\TheTitle}{Nematic Liquid Crystals in Cuboids}
\newcommand{\TheAuthors}{B. Shi, Y. Han, A. Majumdar, L. Zhang}

\headers{\TheTitle}{\TheAuthors}
\title{{\TheTitle}\thanks{This work was partially supported by National Natural Science Foundation of China (Grant No. 12225102, T2321001, 12050002, 12288101), Newton Advanced Fellowship, Leverhulme Research Project Grant RPG-2021-401, EPSRC Grant Number EP/R014604/1, the Humboldt Foundation, and a University of Strathclyde New Professors Fund.}}

\author{Baoming Shi\thanks{School of Mathematical Sciences, Peking University, Beijing 100871, China (\email{ming123@stu.pku.edu.cn}).}
\and  Yucen Han\thanks{Department of Mathematics and Statistics, University of Strathclyde, G1 1XQ, UK (\email{yucen.han@strath.ac.uk}).}
\and  Apala Majumdar\thanks{Department of Mathematics and Statistics, University of Strathclyde, G1 1XQ, UK (\email{apala.majumdar@strath.ac.uk}).}
\and Lei Zhang\thanks{Beijing International Center for Mathematical Research, Center for Quantitative Biology, Center for Machine Learning Research, Peking University, Beijing 100871, China (\email{zhangl@math.pku.edu.cn}).}
}
\ifpdf
\hypersetup{
  pdftitle={\TheTitle},
  pdfauthor={\TheAuthors}
}
\fi
\renewcommand{\TheTitle}{Multistability for Nematic Liquid Crystals in Cuboids with Degenerate Planar Boundary Conditions}

\begin{document}
\maketitle

\begin{abstract}
 We study nematic configurations within three-dimensional (3D) cuboids, with planar degenerate boundary conditions on the cuboid faces, in the Landau-de Gennes framework. There are two geometry-dependent variables: the edge length of the square cross-section, $\lambda$, and the parameter $h$, which is a measure of the cuboid height. Theoretically, we prove the existence and uniqueness of the global minimiser with a small enough cuboid size. We develop a new numerical scheme for the high-index saddle dynamics to deal with the surface energies. We report on a plethora of (meta)stable states, and their dependence on $h$ and $\lambda$, and in particular, how the 3D states are connected with their two-dimensional counterparts on squares and rectangles. Notably, we find families of almost uniaxial stable states constructed from the topological classification of tangent unit-vector fields and study transition pathways between them. We also provide a phase diagram of competing (meta)stable states, as a function of $\lambda$ and $h$.
 \end{abstract}



\section{Introduction}
\label{sec:intro}
Liquid crystals (LCs) are mesophases, that are intermediate in character between the solid and liquid phases of matter \cite{de1993physics,wang2021modeling}. There are different types of LCs, of which nematic liquid crystals (NLCs) are the simplest and most commonly used in science and technology. NLCs combine fluidity with the directionality of solids i.e. NLCs have long-range orientational order with distinguished directions of preferred molecular alignment, referred to as nematic ‘‘directors’’ in the literature \cite{de1993physics}. The intrinsic anisotropy makes NLCs highly sensitive to external stimuli e.g. electric fields, incident light, temperature, stress and surface effects. Indeed, the exceptional properties of NLCs make them the working material of choice for the multi-billion dollar liquid crystal display (LCD) industry, and NLC applications now extend to soft robotics, biomimetic materials, sensors and light modulators \cite{LAGERWALL20121387,bisoyi2021liquid, edwards2020interaction,loussert2013manipulating}. 

NLC applications can depend quite strongly on anchoring conditions or boundary conditions, i.e. the coupling of the NLC molecules to surfaces can determine the nematic director profiles on the surfaces \cite{de1993physics}. The anchoring conditions are typically either planar degenerate/tangential, for which the director is tangent to the surface or in the plane of the surface, or homeotropic/normal for which the director is orthogonal to the surface. For example, in \cite{tsakonas2007multistable}, the authors report an NLC-filled 3D array of square or rectangular wells, such that the well surfaces are treated to induce tangent boundary conditions. The tangent boundary conditions induce bistability, i.e. the wells can support two optically contrasting stable NLC states, without any external fields. 
In \cite{noh2020PRR}, the authors study NLC shells and the shell surfaces are treated with a polymer, such that the boundary conditions can be dynamically tuned from tangential to normal, as the shells undergo a heating transition. The change in the boundary conditions manifests in the experimentally recorded optical images. In \cite{caimi2021surface}, the authors explore the surface alignment of the ferroelectric nematic phase
by testing different rubbed and unrubbed substrates that differ in coupling strength and anchoring
orientation and find a variety of behaviours – in terms of nematic orientation, topological defects and
electric field response. In \cite{Chung_2003}, the authors study NLC molecular orientations on a doubly treated substrate, with different surface alignments on the top and bottom surfaces, and use the simulation results to estimate surface anchoring strengths.

In previous work \cite{shi2023hierarchies}, we study NLCs confined to a 3D cuboid with Dirichlet/fixed tangent boundary conditions on the lateral surfaces of the cuboid, and with either Neumann (natural) boundary conditions or fixed Dirichlet boundary conditions on the top and bottom cuboid surfaces. In both cases, we work within the celebrated continuum Landau-de Gennes (LdG) framework and model physically relevant configurations as minimisers of an appropriately defined LdG energy, or stable solutions of the associated Euler-Lagrange equations which are a system of nonlinear and coupled partial differential equations \cite{majumdar2010equilibrium}. 
With Neumann boundary conditions on the top and bottom cuboid surfaces, we find 3D $z$-invariant solutions, mixed solutions (also reported in \cite{canevari_majumdar_wang_harris}), and multi-layer solutions with mixed solutions stacked on the top of each other \cite{shi2023hierarchies}. In \cite{han2022prism}, we fix Dirichlet conditions on the top and bottom surfaces in terms of appropriately defined stable solutions of the LdG Euler-Lagrange equations on square/rectangular domains, often referred to as reduced LdG (rLdG) solutions in two-dimensional (2D) settings \cite{han2020reduced}. This choice of Dirichlet conditions is special, and we use our wealth of knowledge of pathways between stable rLdG solutions on squares/rectangles, to construct 3D solutions on a cuboid. Namely, in some cases, if $A \to B \to C$ is a pathway on the rLdG solution landscape on a square domain, where $A, B, C$ denote rLdG solutions (critical points of an appropriately defined rLdG energy), then we can construct a 3D critical point on a cuboid with fixed Dirichlet conditions, corresponding to the $A$, $C$ solutions on the top and bottom surfaces, with a $B$-profile located at the middle of the cuboid. In other words, we can construct 3D critical points of the LdG energy by stacking rLdG solutions on square domains, on top of each other. In this paper, we build on our previous work and relax the fixed Dirichlet tangent conditions on the lateral surfaces of the cuboid. Rather, we work with surface energies on all six surfaces, that enforce planar degenerate or tangent anchoring on all cuboid faces. 
In this case, the nematic director is only coerced to be in the plane of the face, without a fixed direction in contrast to the Dirichlet fixed boundary conditions on the lateral surfaces in \cite{han2022prism,shi2022nematic}. This certainly allows for more freedom on the lateral faces, and expands the corresponding solution landscapes. 

In a batch of papers \cite{han2020reduced,kralj2014order, robinson2017molecular,yin2020construction, fang2020surface}, the authors study the rLdG model on 2D square and rectangular domains. The rLdG model can be viewed as a restriction of the LdG model to 2D domains with tangent boundary conditions; more details are given in the next section. A rectangular domain is characterized by an edge length, $\lambda$, and an aspect ratio, $b$. For $b=1$, it is known that for $\lambda$ small enough, the Well Order Reconstruction Solution (WORS) is the unique stable rLdG solution on a square domain with tangent boundary conditions.  The WORS is distinguished by two orthogonal defect lines along the two square diagonals, and as $\lambda$ increases, the WORS bifurcates to BD solutions with parallel line defects along opposite edges, and then the stable diagonal (D) and rotated (R) solutions (first reported in \cite{tsakonas2007multistable}). If $b \neq 1$, the competing stable solutions are the BD, D and R solutions. 

In this paper, we have a 3D cuboid with two geometry-dependent variables: the edge length of the square cross-section denoted by $\lambda$, and the cuboid height or the aspect ratio denoted by $h$. The anchoring strength, $W$, is yet another parameter, a measure of how strongly the tangent boundary conditions are enforced on the cuboid faces. We prove that LdG energy minimisers strictly respect tangent boundary conditions in the $W\to\infty$ limit, and we work with a large $W$ throughout the manuscript. We design a numerical scheme to deal with the ill-conditioning posed by the surface energies in some regimes. We then, numerically investigate the LdG solution landscape on cuboids for different geometrical regimes, defined by $\lambda$ and $h$, i.e. when the $z$-edge is longer or shorter than the $x$,$y$-edges and when we have a cube of all equal edge lengths. Notably, we generate exotic stable solutions in the LdG framework, which exhibit rLdG solutions on the six cuboid faces e.g. WORS-WORS-WORS solution with a WORS-type profile on all six faces and exotic defect structures on the surfaces and in the cuboid interior, BD-BD-BD type solutions with BD-type profiles on all six faces, WORS-BD-BD profiles with a mix of WORS and BD-profiles on the six faces, D-BD-BD or R-BD-BD solutions with a mix of D, R and BD-type profiles on the cuboid faces. Last but not the least, we use the topological methods from \cite{robbins2004classification} to numerically find at least six different stable solutions, with D and R-type profiles on all cuboid faces, which we believe to be globally stable for large $\lambda$ and large $h$.  
We also study transition pathways between competing stable solutions, and find multiple possibilities, all dictated by the vertex defects moving along either cuboid edges or face diagonals. We summarise our findings in a phase diagram, as a function of $h$ and $\lambda$, which summarises our numerical results.

In Sec. 2, we briefly review the LdG theory and illustrate our setup. In Sec. 3, we introduce the high-index saddle dynamics and illustrate the new numerical scheme to accelerate and stabilize the numerical computations with surface energies. In Sec. 4, we present numerical results for small $h$, large $h$ and $h = 1$. We finally present our conclusions and discussions in Sec. 5.

\section{The Landau--de Gennes theory}

We work within the celebrated LdG theory,  which is the most general continuum theory for NLCs. The LdG theory describes the state of NLC ordering by a macroscopic order parameter, the $\Q$-tensor, that distinguishes NLCs from isotropic liquids \cite{de1993physics}. 
Mathematically, the $\Q$-tensor is a symmetric and traceless $3\times3$ matrix as shown below:
\begin{equation}\label{eq:5-degree}
    \Q = \begin{pmatrix}
        Q_1 & Q_3 & Q_4\\
        Q_3 & Q_2 &  Q_5\\
        Q_4 & Q_5 &  -Q_1-Q_2
    \end{pmatrix}.
\end{equation}
 From the spectral decomposition theorem, we can write the $\Q$-tensor as
$$
\Q = \sum_{i=1}^{3} \lambda_{i} \mathbf{e}_i \otimes \mathbf{e}_i,
$$
where $\left\{ \mathbf{e}_1, \mathbf{e}_2, \mathbf{e}_3 \right\}$ are the eigenvectors of the $\Q$-tensor and $\lambda_1\leqslant\lambda_2\leqslant\lambda_3$ are the associated eigenvalues respectively, subject to $\sum_{i=1}^{3}\lambda_i = 0$. The eigenvectors model the preferred directions of spatially averaged local molecular alignment or \emph{the nematic directors}, and the eigenvalues are a measure of the orientational order about these directions. A $\Q$-tensor is said to be isotropic if $\Q=\mathbf{0}$, uniaxial if $\Q$ has a pair of repeated non-zero eigenvalues, and biaxial if $\Q$ has three distinct eigenvalues \cite{de1993physics,mottram2014introduction}. A uniaxial NLC phase has a single distinguished direction of averaged molecular alignment (modelled by the eigenvector with the non-degenerate eigenvalue), such that all directions perpendicular to the uniaxial director are physically equivalent. A biaxial phase has a primary and secondary nematic director. 

The LdG theory is a variational theory, and the physically observable configurations are modelled by minimisers of an appropriately defined LdG free energy \cite{de1993physics}. We work with a simple form of the LdG free energy:
\begin{equation} \label{eq:LdG}
	E_{LdG}[\Q]: = \int_{V} \left[\frac{L}{2}\left| \nabla \Q \right|^2 + f_B\left( \Q \right)\right]\mathrm{d}V, 
\end{equation}
where the first term in the integrand is the elastic energy density that penalizes spatial inhomogeneities, $\left| \nabla \Q \right|^2:=\frac{\partial Q_{ij}}{\partial r_k}\frac{\partial Q_{ij}}{\partial r_k}$, $i,j,k = 1,2,3$ and the second term is the thermotropic potential that dictates the preferred NLC phase as a function of temperature, 
\begin{equation}\label{f_B}
	f_B(\Q): = \frac{A}{2}\mathrm{tr} \Q^2 - \frac{B}{3} \mathrm{tr} \Q^3 + \frac{C}{4} (\mathrm{tr} \Q^2)^2-f_{B,0},
\end{equation}
where $\mathrm{tr} \Q = Q_{ij} Q_{ij}$, $i,j = 1,2,3$ and we use the Einstein summation convention throughout the paper.

Here, $L>0$ is a material-dependent elastic constant, $A=\alpha (T - T^*)$ is the rescaled temperature, with $\alpha>0$ and $T^*$ is a characteristic liquid crystal temperature; $B, C>0$ are material-dependent bulk constants. The minimisers of $f_B$ depend on $A$ and determine the NLC phase for spatially homogeneous samples. For $A<0$, the minimisers of $f_B$ constitute a continuum of uniaxial $\Q$-tensors defined by
\[
\mathcal{N} = \left\{ \Q = s_+\left(\mathbf{n}\otimes \mathbf{n} - \frac{\mathbf{I}}{3} \right) \right\},
\]
where
$$ s_+ = \frac{B + \sqrt{B^2 - 24 AC}}{4C}, $$
and
$\mathbf{n}$ is an arbitrary unit vector field that models the uniaxial director. The constant, $f_{B,0}=\frac{A}{3}s_+^2-\frac{2B}{27}s_+^3+\frac{C}{9}s_+^4$ \cite{majumdar2010equilibrium}, is added to ensure a non-negative energy density.

Our working domain is a cuboid $V=\left[ -\lambda,\lambda\right]^2 \times \left[-\lambda h,\lambda h \right]$ where $\lambda$ is edge-length of the 2D square cross-section and $h>0$ is a measure of the height. In a batch of papers \cite{han2020reduced, han2021elastic}, the authors work with the rLdG model valid for 2D domains or the thin-film limit of 3D domains (the $h\to 0$ limit of the cuboid $V$ above) and certain boundary conditions, for which the energy-minimising or physically relevant $\Q$-tensors have constant $\frac{Q_1 + Q_2}{2}$, and $Q_4 = Q_5 = 0$, so that the rLdG order parameter has only two degrees of freedom, $\frac{Q_1-Q_2}{2}$ and $Q_3$ respectively. The rLdG order parameter is $z$-invariant and only describes in-plane nematic ordering in a 2D domain, which cannot work for truly 3D scenarios. However, at a special characteristic low-temperature, $A = -B^2/ 3C$, there exists a branch of critical points of \eqref{eq:LdG} of the form
\[
\Q = \P - \frac{B}{3C}(2 \zhat \otimes \zhat - \xhat \otimes \xhat - \yhat \otimes \yhat )
\] where $\P$ is the rLdG order parameter with two degrees of freedom; $\xhat$, $\yhat$, $\zhat$ are the unit-vectors in the co-ordinate directions respectively and $\P$ is a critical point of an appropriately defined rLdG free energy, and this branch exists for all $h>0$. 
In this paper, we use the same parameter values as in \cite{han2020reduced, han2021elastic} to facilitate comparisons between the 2D solution landscapes and their relevance for 3D problems, i.e. we fix $B=0.64\times10^4 \text{Nm}^{-2}$,  $C=0.35\times10^4 \text{Nm}^{-2}$, and $L=4\times10^{-11} \text{N}$ (material constants for the representative NLC material MBBA) \cite{majumdar2010equilibrium, wojtowicz1975introduction} and work at the special temperature $A=-B^2/3C$,  as in \cite{shi2023hierarchies,robinson2017molecular,yin2020construction}.

By rescaling the system according to $(\bar{x},\bar{y},\bar{z}) = (\frac{x}{\lambda},\frac{y}{\lambda},\frac{z}{\lambda})$, $\bar{\lambda}^2=\frac{2C\lambda^2}{L}$ and dropping the bars in subsequent discussions (so that all results are in terms of dimensionless variables), the non-dimensionalised LdG free energy is given by,
\begin{equation}
	E_{LdG}[\Q]:=\int_{V} \left[\dfrac{1}{2}\left|\nabla \Q \right|^2 +\lambda^2 \left(\frac{A}{4C}\text{tr}\Q^2-\frac{B}{6C}\text{tr}\Q^3+\frac{1}{8}(\text{tr}\Q^2)^2-\frac{f_{B,0}}{2C}
	\right) \right]\mathrm{d}V.
	\label{energy_b}
\end{equation}
The normalized domain is $V=\Omega\times \left[-h,h \right]$, $\Omega=\left[ -1,1\right]^2$ is the two-dimensional cross-section of the cuboid, and $\lambda^2$ describes the cross-sectional size. We use the biaxiality parameter, 
\begin{equation}
  \beta^2=1-6\text{tr}(\Q^3)^2/(\text{tr}(\Q^2))^3,
  \label{eq: biaxiality parameter}
\end{equation}
to visualize defects, since defects are typically surrounded by regions with high biaxiality \cite{PhysRevE.60.1858,Canevari}. We have $0\leqslant \beta^2 \leqslant 1$, and $\beta^2=0$ if and only if $\Q$ is uniaxial or isotropic \cite{landau2010beta}.

We impose surface energies on all six cuboid faces: $\partial V_{tb},\partial V_{fb},\partial V_{lr}$, which implicitly enforce planar degenerate anchoring on all six faces, i.e. coerce the leading eigenvector (with the largest positive eigenvalue) of the $\Q$-tensor to be in the plane of the face. The total free energy is then given by,

\begin{equation}
    E_{\omega}[\Q]:=E_{LdG}(\Q)+\omega E_{bc}(\Q)=\int_{V} \left[\dfrac{1}{2}\left|\nabla \Q \right|^2 +\lambda^2 f_b\right]\mathrm{d}V+\omega \int_{\partial V} \Vert\Q\Vec{\nu}+\frac{s_+}{3}\Vec{\nu}\Vert^2 \mathrm{d}S,
  \label{eq: energy total}
  \end{equation}
where $\Vec{\nu}$ is the outer normal vector of $\partial V$, $\omega=\frac{\lambda W}{\sqrt{2CL}}$ is the non-dimensionalised anchoring strength, and $W$ is the surface anchoring parameter, which is a measure of how strongly the boundary conditions are enforced on $\partial V$ \cite{ravnik2009landau}. If $W = 0$, then we have natural boundary conditions on all six surfaces. The surface energy, $E_{bc}$, vanishes if and only if  $\Vec{\nu}$ is an eigenvector with constant eigenvalue $-s_+/3$, so that the leading eigenvector is in the plane of the cuboid face, but our choice of $E_{bc}$ does not fix or prefer a specific orientation/direction for the leading planar eigenvector, i.e. the leading eigenvector is free to rotate in the plane of the cuboid faces. In fact, $E_{bc} = 0$ precisely when $\Q$ only has two degrees of freedom on the faces of $V$, as in the rLdG model, and hence, it is reasonable to expect that the surface profiles of the minimisers of \eqref{eq: energy total} are critical points/minimisers of the rLdG free energy in \cite{han2020reduced}. 

There are discontinuities of the outer normal vector on the edges. For the following analytic results, we take the domain $V$ to be a ``smoothed cuboid'' (as shown in Fig. \ref{figure 1}) with a small fixed truncation $\epsilon$. For $\epsilon$ small enough, this assumption will not affect our qualitative predictions. The admissible $\Q$-tensors belong to the space 
\begin{equation}
W^{1,2}_{V,\S_0}=\{\Q \in \S_0|\Q \in W_V^{1,2}\},
\end{equation}
where 
\begin{equation}
\begin{aligned}
\S_0:=\{\Q \in \mathbb{R}^{3\times 3}: \Q_{ij}=\Q_{ji},\Q_{ii}=0\},\\
W_V^{1,2}=\left\{u:  \int_V \left(|u|^2+|\nabla u|^2 \right)\mathrm {d}A<\infty \right\}.
\end{aligned}
\end{equation}
In the ensuing propositions and lemmas, we prove some generic existence, uniqueness and maximum-principle type of results for critical points of \eqref{eq: energy total} on the smoothed cuboid.

\begin{figure}[H]
  \centering
  \includegraphics[width=.3\textwidth]{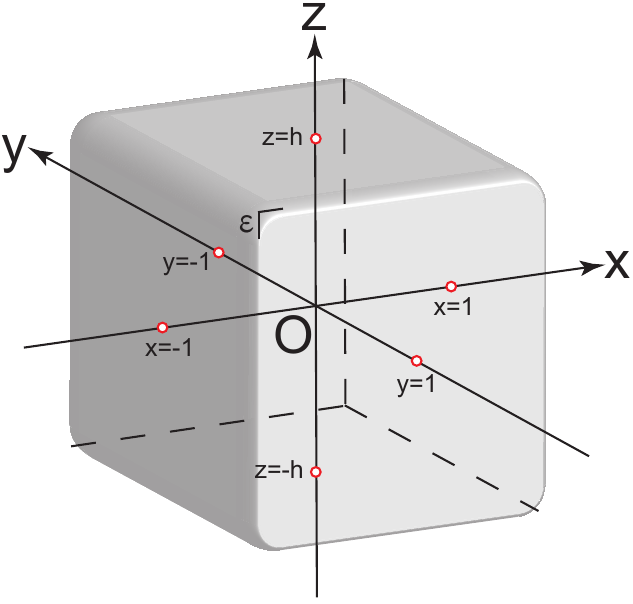}
  \caption{The smoothed cuboid $V$ with a small fixed truncation $\epsilon$.}
  \label{figure 1}
\end{figure}

\begin{lemma} (Poincare-Friedrichs inequality) \label{lemma: Poincare}
  For any $\partial V_0\subseteq \partial V$ with a positive 2D Lebesgue measure, there exist two constants $\gamma_1\geqslant \gamma_0>0$ depending on $h$ subject to,
  \begin{equation}\label{PF}
    \gamma_0(h)\Vert u\Vert_{W_V^{1,2}}^2\leqslant \Vert u \Vert_{L_{\partial V_0}^2}^2+\Vert\nabla u \Vert_{L_V^2}^2\leqslant \gamma_1(h)\Vert u \Vert_{W_V^{1,2}}^2, \forall u\in W_V^{1,2}.
  \end{equation}
\end{lemma}

\begin{proof}
The proof follows the same paradigm as in A.9 Theorem in \cite{struwe2000variational}. A brief proof is provided for completeness. The second inequality can be directly derived from the embedding theorem \cite{adams2003sobolev}. We prove the first inequality in \eqref{PF} by contradiction.

    Assuming that the first inequality does not hold with any positive $\gamma_0$, i.e. there exists a sequence $\{u_k\} \in W^{1,2}_V$ which satisfies $\Vert u_k\Vert_{W_V^{1,2}}^2=1$ and,
    \begin{equation}
        \Vert u_k \Vert_{L_{\partial V_0}^2}^2+\Vert\nabla u_k \Vert_{L_V^2}^2\rightarrow 0, k\rightarrow \infty.
        \label{eq: poincare1}
    \end{equation}
The bounded sequence in $W^{1,2}_V$ is quasi-weakly compact, and from the $W^{1,2}_V$ compact embedding in $L^2_V$ \cite{adams2003sobolev}, there exists a sub-sequence $\{u_{k_j}\}$ and $u_1\in W^{1,2}_V, u_2\in L^2_V$, such that
\begin{equation}
    u_{k_j}\rightharpoonup u_1, \text{in } W^{1,2}_V, \text{ and }  u_{k_j}\rightarrow u_2, \text{in } L^2_V.
    \label{eq: poincare2}
\end{equation}
From \eqref{eq: poincare1} and \eqref{eq: poincare2}, $\{u_{k_j}\}$ is a Cauchy sequence in $W^{1,2}_V$, and therefore, we have $u_1=u_2=u$ and $\Vert u_{k_j}-u_1 \Vert^2_{W^{1,2}_V}\rightarrow 0$. Together with \eqref{eq: poincare1}, it follows that $u\equiv C_0$ is a constant and \eqref{eq: poincare1} implies that $C_0=0$. However, $\Vert u_k\Vert_{W_V^{1,2}}^2=1$ and hence, $C_0\neq 0$, leading to a contradiction.
\end{proof}

\begin{lemma}\label{lemma: lower-semicontinuous}
    $E_\omega$ is weakly sequentially lower semi-continuous for any $\omega\geqslant0$, i.e.
    \begin{equation}    \liminf_{k\rightarrow \infty}E_{\omega}(\Q_k)\geqslant E_{\omega}(\Q), \Q_k\rightharpoonup \Q \text{ in } W^{1,2}_{V,\S_0}.
    \end{equation}
\end{lemma}
\begin{proof}
    The lower semi-continuity property of $E_{LdG}$ follows directly from the fact the energy density is convex in $\nabla \Q$ \cite{struwe2000variational}.

    From the trace theorem \cite{adams2003sobolev}, there exists a bounded linear operator,
    \begin{equation}
        T: \Q\rightarrow \Q|_{\partial V}, W^{1,2}_{V,\S_0}\rightarrow L^2_{\partial V,\S_0}.
    \end{equation}
We need to show that the surface energy, $E_{bc}(\Q|_{\partial V})=E_{bc}(T(\Q))$, is also weakly sequentially lower semi-continuous on $W^{1,2}_{V,\S_0}$. For any bounded linear operator $F \in (L^{2}_{\partial V,\S_0})^*$ (the dual space of $L^{2}_{\partial V,\S_0}$), $F\circ T$ is a bounded linear operator on $W^{1,2}_{ V,\S_0}$, that is, if $\Q_k \rightharpoonup \Q$, we can have $T(\Q_k) \rightharpoonup T(\Q)$. Further, $E_{bc}(\Q|_{\partial V})$ is weakly sequentially lower semi-continuous on $L^2_{\partial V,\S_0}$, because it is convex and continuously dependent on $\Q|_{\partial V}$ \cite{struwe2000variational}. Thus, if $\Q_k \rightharpoonup \Q$, we can get
\begin{equation}  \liminf_{k\rightarrow \infty}{E_{bc}(T(\Q_k))}\geqslant {E_{bc}(T(\Q))}.
\end{equation}
Finally, the weakly sequentially lower semi-continuity of $E_{\omega}$ follows from
\begin{equation}
\begin{aligned}
\liminf_{k\rightarrow \infty}{E_{\omega}(\Q_k)}&\geqslant \liminf_{k\rightarrow \infty}{E_{LdG}(\Q_k)}+\omega \liminf_{k\rightarrow \infty}{E_{bc}(T(\Q_k))}\\
&\geqslant E_{LdG}(\Q)+\omega E_{bc}(T(\Q))=E_\omega(\Q).  
\end{aligned}
\end{equation}

\end{proof}

\begin{proposition}\label{pro:1}
 For any $\omega>0$, $E_{\omega}$ has a global minimiser $\Q(\omega)$ in $W^{1,2}_{V,\S_0}$ which satisfies 
  \begin{equation}
    \begin{aligned}
    E_{LdG}(\Q(\omega_{n_2}))\geqslant E_{LdG}(\Q(\omega_{n_1}))&, \text{ } E_{bc}(\Q(\omega_{n_2}))\leqslant E_{bc}(\Q(\omega_{n_1})),\text{ if } \omega_{n_2}> \omega_{n_1},\\
   & \lim_{\omega_n\rightarrow \infty} E_{bc}(\Q(\omega_n))=0.
    \end{aligned}
  \end{equation}
  If $\Q(\omega_n)$ has a subsequence such that $\Q(\omega_{n_k})\rightarrow \Q_\infty$ strongly in $W^{1,2}_{V,\S_0}$ as $\omega_n \rightarrow \infty$, then $\Q_\infty \in W^{1,2}_{V,\S_0}$ is a global minimiser of the LdG energy, which satisfies perfect planar degenerate surface anchoring, i.e. $E_{bc}(\Q_\infty)=0$.
\end{proposition}

\begin{proof}
The energy density in $E_\omega$ is non-negative, i.e. it is bounded from below, and $E_\omega$ is weakly sequentially lower semi-continuous (Lemma \ref{lemma: lower-semicontinuous}). We only need a coerciveness estimate, and the existence of a global minimiser follows from the direct methods in the calculus of variations. 
  
  By taking $\partial V_0=\partial V_{lr}$ in Lemma \ref{lemma: Poincare}, we have
\begin{equation}
  \begin{aligned}
    \Vert Q_1 \Vert_{W_V^{1,2}}^2 \leqslant & \frac{1}{\gamma_0} \left( \int_{\partial\V_{lr}}Q_1^2\mathrm{d}S+\int_V \Vert \nabla Q_1\Vert^2\mathrm{d}V\right)\\
 \leqslant & \frac{1}{\gamma_0}\left(2\int_{\partial\V_{lr}}\left[\left(Q_1+\frac{s_+}{3}\right)^2+\left(\frac{s_+}{3}\right)^2\right] \mathrm{d}S+\int_V \Vert \nabla Q_1\Vert^2\mathrm{d}V\right)
 \\
 \leqslant & \frac{2}{\gamma_0} (E_{bc}+\bar{C}+E_{elastic}),
\end{aligned}
\label{poincare}
\end{equation}
where $\bar{C}$ is a constant which only depends on $s_+$ and $V$. Subsequently, we have
\begin{equation}
  \begin{aligned}
  E_\omega(\Q)\geqslant & E_{elastic}+\omega E_{bc}
  \geqslant \min(1,\omega)(E_{elastic}+ E_{bc}) \\
  \geqslant & \frac{\min(1,\omega)\gamma_0}{2} \Vert Q_1 \Vert_{W_V^{1,2}}^2-\min(1,\omega)\bar{C}.
\end{aligned}
\label{eq: coercive}
\end{equation}
The same arguments apply to $\partial V_0 = \partial V_{tb}$, $\partial V_{fb}$ or $\partial V_{lr}$ and $Q_i,i=2,\cdots,5$, for different choices of the constants and hence, 
$E_\omega(\Q)$ is coercive.
 
Let $\Q(\omega_n)\in W^{1,2}_{V,\S_0}$ be a global minimiser of $E_{\omega_n}(\Q)$ and $\omega_n \rightarrow \infty$, 
\begin{equation}
\begin{aligned}
E_{\omega_{n_1}}(\Q(\omega_{n_2}))&=E_{LdG}(\Q(\omega_{n_2}))+\omega_{n_1} E_{bc}(\Q(\omega_{n_2}))\\
  &\geqslant E_{LdG}(\Q(\omega_{n_1})) + \omega_{n_1} E_{bc}(\Q(\omega_{n_1}))=E_{\omega_{n_1}}(\Q(\omega_{n_1})),
  \end{aligned}
    \label{eq: omega1omega2 1}
\end{equation}
\begin{equation}
\begin{aligned}
E_{\omega_{n_2}}(\Q(\omega_{n_1}))&=E_{LdG}(\Q(\omega_{n_1}))+\omega_{n_2} E_{bc}(\Q(\omega_{n_1}))\\
  &\geqslant E_{LdG}(\Q(\omega_{n_2})) + \omega_{n_2} E_{bc}(\Q(\omega_{n_2}))=E_{\omega_{n_2}}(\Q(\omega_{n_2})).
  \end{aligned}
  \label{eq: omega1omega2 2}
\end{equation}
Adding both sides of the inequalities \eqref{eq: omega1omega2 1} and \eqref{eq: omega1omega2 2}, we have
\begin{equation}
    (\omega_{n_2}-\omega_{n_1})E_{bc}(\Q(\omega_{n_1}))\geqslant (\omega_{n_2}-\omega_{n_1})E_{bc}(\Q(\omega_{n_2})),
\end{equation}
and hence,
\begin{equation}
E_{bc}(\Q(\omega_{n_2}))\leqslant E_{bc}(\Q(\omega_{n_1})),\text{ if } \omega_{n_2}> \omega_{n_1},
\label{eq: bc order}
\end{equation}
Substituting \eqref{eq: bc order} into \eqref{eq: omega1omega2 1}, it follows that
\begin{equation}
    E_{LdG}(\Q(\omega_{n_2}))\geqslant E_{LdG}(\Q(\omega_{n_1})),\text{ if } \omega_{n_2}> \omega_{n_1}. 
\end{equation}
Take any $\hat{\Q} \in W_V^{1,2}$ which satisfies $E_{bc}(\hat{\Q})=0$, and then we have the following sequence of inequalities:
\begin{equation}
  E_{LdG}(\hat{\Q})= E_{\omega_n}(\hat{\Q})\geqslant E_{\omega_n}(\Q(\omega_n)) \geqslant \omega_n E_{bc}(\Q(\omega_n))\geqslant 0,
\end{equation}
so that, 
\begin{equation}
  \lim_{\omega_n\rightarrow \infty} E_{bc}(\Q(\omega_n))=0.
\label{eq: limit boundary energy}
\end{equation}
If $\Q(\omega_n)$ has a subsequence, $\Q(\omega_{n_k})\rightarrow\Q_{\infty}$ strongly in $W^{1,2}_{V,\S_0}$ as $\omega_{n_k}\rightarrow \infty$, then $E_{bc}(\Q_{\infty})=0$ from \eqref{eq: limit boundary energy}. Thus the limit, $\Q_\infty \in W^{1,2}_{V,\S_0}$ is a global minimiser of the LdG energy in the admissible space. 
\end{proof}

\begin{lemma}\label{lemma: EL}
  The critical points of the functional \eqref{eq: energy total} in $W_{V,\S_0}^{1,2}$ satisfy the Euler-Lagrange equation,
  \begin{equation}
    \Delta \Q=\lambda^2\left(\frac{A}{2C}\Q-\frac{B}{2C}\left(\Q^2-\frac{\mathrm{tr}(\Q^2)}{3}\I \right)+\frac{1}{2}\mathrm{tr}(\Q^2)\Q\right), r \in V\setminus \partial V
    \label{EL_b}
  \end{equation}
  with the boundary condition
  \begin{equation}
      \partial_{\Vec{\nu}} \Q+\omega \left(\vec{\nu} \vec{\nu}^\top \Q+\Q\vec{\nu}\vec{\nu}^\top+\frac{2s_+}{3}\vec{\nu}\vec{\nu}^\top -\left(\frac{2}{3}\vec{\nu}^\top\Q\vec{\nu}+\frac{2s_+}{9}\right)\I\right)=0, r\in \partial V.
      \label{eq: boundary condition implicit}
  \end{equation}
\end{lemma}

\begin{proof}
  Let $\Q \in W_{V,\S_0}^{1,2}$  be a critical point of \eqref{eq: energy total}
, and $\P \in W_{V,\S_0}^{1,2}$ be a perturbation.
We compute the first variation of \eqref{eq: energy total} as,
\begin{equation}\label{eq: E-L 1}
  \begin{aligned}
  0&=\frac{d}{dt|_{t=0}}E_\omega (\Q+t\P)\\
  &=\int_V \left( -\Delta \Q+\lambda^2\left(\frac{A}{2C}\Q-\frac{B}{2C}\left(\Q^2-\frac{\mathrm{tr}(\Q^2)}{3}\I \right)+\frac{1}{2}\mathrm{tr}(\Q^2)\Q\right)\right)\cdot \P\mathrm{d}V\\
  &+\int_{\partial V} \partial_{\vec{\nu}} \Q\cdot \P + \omega \left(\vec{\nu}^\top \Q\P\vec{\nu}+\vec{\nu}^\top \P\Q\vec{\nu}+\frac{2s_+}{3} \vec{\nu}^\top \P\vec{\nu} \right)\mathrm{d}S.
  \end{aligned}
\end{equation}
Recalling that $\mathrm{tr}\P = 0, \P\cdot \I=0$, we obtain
\begin{equation}\label{eq: E-L 2}
  \begin{aligned}
  &\int_{\partial V} \left(\vec{\nu}^\top \Q\P\vec{\nu}+\vec{\nu}^\top \P\Q\vec{\nu}+\frac{2s_+}{3} \vec{\nu}^\top \P\vec{\nu} \right)\mathrm{d}S\\
  &=\int_{\partial V} \left(\vec{\nu} \vec{\nu}^\top \Q+\Q\vec{\nu}\vec{\nu}^\top+\frac{2s_+}{3}\vec{\nu}\vec{\nu}^\top\right)\cdot \P \mathrm{d}S\\
  &=\int_{\partial V} \left(\vec{\nu} \vec{\nu}^\top \Q+\Q\vec{\nu}\vec{\nu}^\top+\frac{2s_+}{3}\vec{\nu}\vec{\nu}^\top -\left(\frac{2}{3}\vec{\nu}^\top\Q\vec{\nu}+\frac{2s_+}{9}\right)\I\right)\cdot \P \mathrm{d}S,
  \end{aligned}
\end{equation}
where $\left(\frac{2}{3}\vec{\nu}^\top\Q\vec{\nu}+\frac{2s_+}{9}\right)\I$ is the Lagrange multiplier associated with the tracelessness constraint. The Euler-Lagrange equations follow from \eqref{eq: E-L 1} and \eqref{eq: E-L 2}, by the fundamental theorem of the calculus of variations.
\end{proof}

\begin{remark}
 The boundary energy does not change the Euler-Lagrange equation \eqref{EL_b}, so that we can improve the regularity of critical point $\Q \in W_{V,S_0}^{1,2}$ as Proposition 13 in \cite{majumdar2010landau}. An argument based on elliptic regularity in Theorem 3 of \cite{dang1992saddle} or Proposition 3.12 of \cite{canevari_majumdar_wang_harris} shows that the critical points $\Q \in W_{V,S_0}^{1,2}$ of $E_\omega(\Q)$ are classical solutions of \eqref{EL_b}.
\end{remark}

\begin{proposition} \label{prop:maximum}
  There exists a constant $M(A,B,C)$ such that the critical point of $E_\omega(\Q)$ satisfy the inequality,
  \begin{equation}
    \Vert \Q\Vert_{L_{V}^\infty}=\text{ess sup}_{r\in V}\left|\Q(r)\right|\leqslant M,
  \end{equation}
  where $| \Q | = \sqrt{\mathrm{tr} \Q^2}$.
\end{proposition}
\begin{proof}
  The proof is analogous to Lemma 3.11 in \cite{canevari_majumdar_wang_harris}. Let $\P=\Q+\frac{s_+}{2}\tilde{\Q}$, where $\tilde{\Q}\in C^2(V)$ is a fixed auxiliary function and satisfies $$\tilde{\Q}(r)=\vec{\nu}  \vec{\nu}^\top, \partial_{\vec{\nu}} \tilde{\Q}(r)=0, r\in \partial V.$$
The existence of $\tilde{\Q}$ follows from the smoothness of $\partial V$.  From Lemma \ref{lemma: EL},  for $r \in \partial V$, we have
  \begin{equation}
    \begin{aligned}
  -\partial_{\vec{\nu}}(|\P|^2/2)&=-\partial_{\vec{\nu}} \P \cdot \P=-\partial_{\vec{\nu}} \Q \cdot \P\\
  &=\omega \left(\vec{\nu} \vec{\nu}^\top \Q+\Q\vec{\nu}\vec{\nu}^\top+\frac{2s_+}{3}\vec{\nu}\vec{\nu}^\top -\left(\frac{2}{3}\vec{\nu}\Q\vec{\nu}+\frac{2s_+}{9}\right)\I\right)\cdot \left(\Q+\frac{s_+}{2}\vec{\nu}\vec{\nu}^\top\right)\\
  &=2\omega \left(\vec{\nu}^\top \left(\Q+\frac{s_+}{3}\I\right)^2\vec{\nu}\right) \geqslant 0,
  \end{aligned}
    \label{eq: maximum principle boundary}
  \end{equation}
and for $r \in V\setminus \partial V$,
\begin{equation}
\begin{aligned}
  \Delta (|\P|^2/2)&=\Delta \P \cdot \P + |\nabla \P|^2\geqslant \Delta \P \cdot \P \\ 
  &\geqslant \Delta \Q \cdot \Q +\frac{s_+}{2}\Delta \Q \cdot \tilde{\Q} +\frac{s_+}{2}\Delta \tilde{\Q} \cdot \Q +\frac{s_+^2}{4} \Delta \tilde{\Q} \cdot \tilde{\Q}= \frac{1}{2}|\Q|^4+\tilde{p}(\Q),
\end{aligned}
  \label{eq: maximum principle laplace}
\end{equation}
where $\tilde{p}(\Q)$ is a third order polynomial of $\Q$. Hence, there exists a positive $M$ (that depends on $A$, $B$, $C$) such that the right-hand side of \eqref{eq: maximum principle laplace} is positive for $|\Q|\geqslant M$. 
By applying the maximum
principle to \eqref{eq: maximum principle boundary} and \eqref{eq: maximum principle laplace}, we deduce that $|\Q|\leqslant M(A,B,C)$.
\end{proof}

\begin{proposition}
   For any $B, C > 0$,  there exists a positive constant $\lambda^*$, depending on $A,B,C,L,W,V$, such that the $E_\omega(\Q)$ has a unique critical point $\hat{\Q} \in W_{V,\S_0}^{1,2}$ for $0\leqslant \lambda<\lambda^*$.
   \end{proposition}

\begin{proof}
  The existence of a global minimiser and hence, a critical point of \eqref{eq: energy total} is proven above. We define the convex set 
  \[S=\{\Q:\Q \in W_{V,\S_0}^{1,2},\Vert \Q\Vert_{L_{V}^\infty}\leqslant M\},\] 
  where $M(A,B,C)$ is defined in Proposition~\ref{prop:maximum}. Then, we prove that $E_{\omega}$ is strictly convex on $S$. For any $\Q,\bar{\Q} \in S$, 
  \begin{equation}
    \begin{aligned}  E_\omega&(\frac{\Q+\bar{\Q}}{2})-\frac{1}{2}E_\omega(\Q)-\frac{1}{2}E_\omega(\bar{\Q})\\
    =&-\frac{1}{8}\Vert \nabla (\Q-\bar{\Q})\Vert_{L_V^2}^2 - \frac{\lambda W}{4\sqrt{2CL}} \int_{\partial V} \Vert \left(\Q-\bar{\Q}\right)\vec{\nu} \Vert^2 \mathrm{d}S \\
    &+\lambda^2 \int_V \left( f_b(\frac{\bar{\Q}+\Q}{2})-\frac{1}{2}f_b(\Q)-\frac{1}{2}f_b(\bar{\Q}) \right)\mathrm{d}V.
    \end{aligned}
  \label{E_convex}.
  \end{equation}
From Lemma \ref{lemma: Poincare}, if $0<\lambda\leqslant \frac{\sqrt{CL}}{\sqrt{2}W}$ s.t. $\frac{\lambda W}{4\sqrt{2CL}}\leq \frac{1}{8}$, then
\begin{equation}
  -\frac{1}{8}\Vert \nabla (\Q-\bar{\Q})\Vert_{L_V^2}^2 - \frac{\lambda W}{4\sqrt{2CL}} \int_{\partial V} \Vert \left(\Q-\bar{\Q}\right)\vec{\nu} \Vert^2 \mathrm{d}S\\
  \leqslant -\frac{\gamma_0(h)\lambda W}{4\sqrt{2CL}} \Vert  \Q-\bar{\Q}\Vert_{W_V^{1,2}}^2.
\end{equation}
Noting that $\Q,\bar{\Q} \in S$, we obtain
\begin{equation}
  \lambda^2 \int_V \left( f_b(\frac{\bar{\Q}+\Q}{2})-\frac{1}{2}f_b(\Q)-\frac{1}{2}f_b(\bar{\Q}) \right)\mathrm{d}V\leqslant \tilde{C}\lambda^2 \Vert  \Q-\bar{\Q}\Vert_{W_V^{1,2}}^2,
\end{equation}
where $\tilde{C}(A,B,C)>0$, and hence 
\begin{equation}
E_\omega(\frac{\Q+\bar{\Q}}{2})-\frac{1}{2}E_\omega(\Q)-\frac{1}{2}E_\omega(\bar{\Q})\leqslant -\frac{\gamma_0(h)\lambda W}{4\sqrt{2CL}} \Vert\Q-\bar{\Q}\Vert_{W_V^{1,2}}^2 + \tilde{C}\lambda^2 \Vert\Q-\bar{\Q}\Vert_{W_V^{1,2}}^2.
\end{equation} 
Thus, $E_\omega$ is strictly convex on $S$, for \begin{equation}
0<\lambda<\lambda^*(A,B,C,L,W)=\min \left(\frac{\sqrt{CL}}{\sqrt{2}W},\frac{\gamma_0(h)W}{4\sqrt{2CL}\tilde{C}(A,B,C)}\right).
\end{equation}
\end{proof}

\begin{figure}[H]
  \centering
  \includegraphics[width=.7\textwidth]{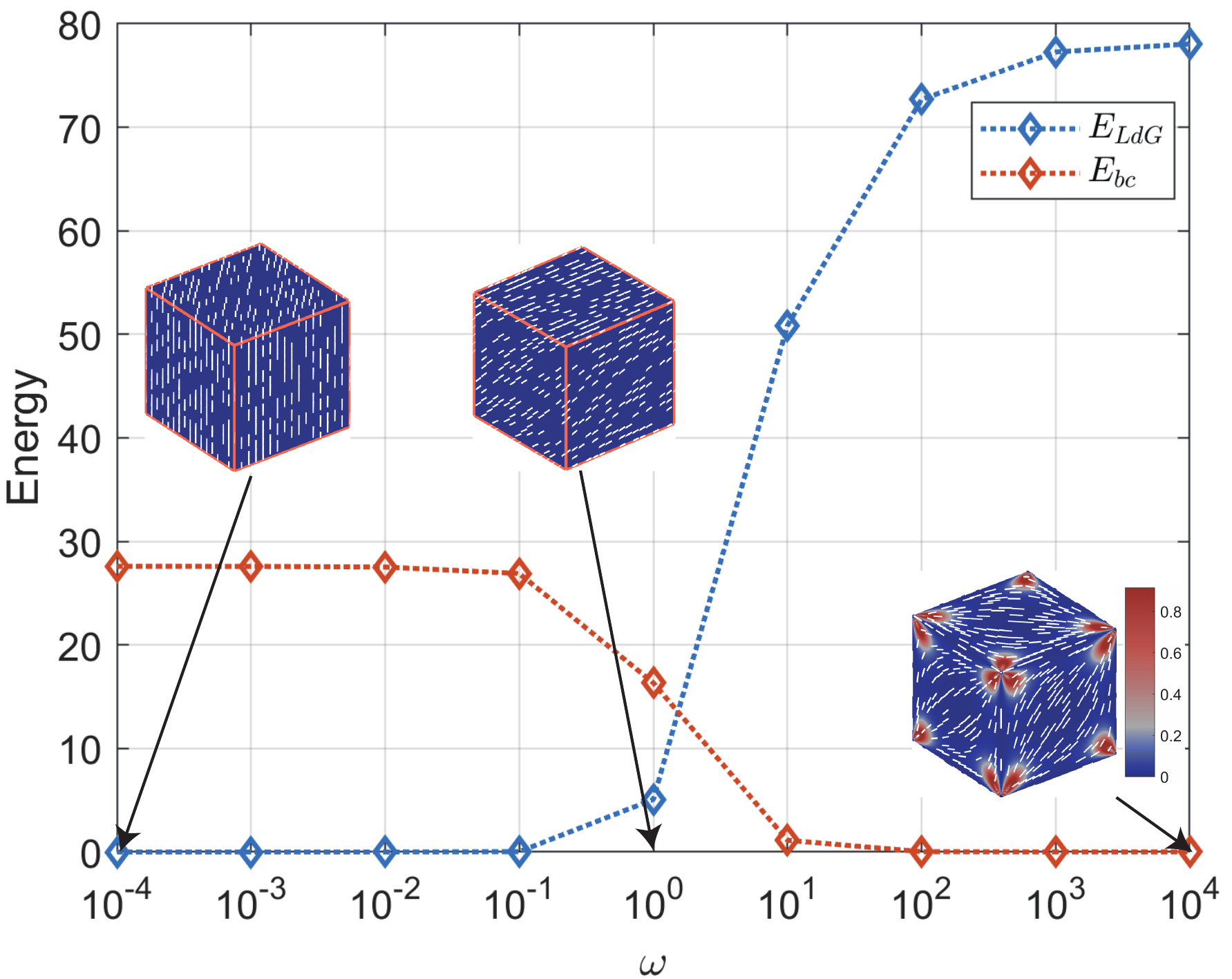}
  \caption{
  The energy plots of $E_{LdG}(\mathbf{Q}(\omega))$ and $E_{bc}(\mathbf{Q}(\omega))$ versus $\omega$ at $h = 1$ and $\lambda^2 = 50$, where $\mathbf{Q}(w)$ is the global minimiser of the free energy $E_{\omega}$. In the configurations for  $\omega = 10^{-4}$, $10^{0}$, and $10^4$, the colour bar labels the biaxiality parameter $\beta^2$ in \eqref{eq: biaxiality parameter} and the white lines represent the directors, i.e. the eigenvector corresponding to the largest eigenvalue of $\mathbf{Q}$. We use this visualization method with the white lines and the colour bar for the following figures. 
For the configurations with $\omega = 10^{-4}$ and $\omega = 10^{0}$, the orange edges are used for better 3D visualization.
}
  \label{figure 2}
\end{figure}

We numerically solve for the global energy minimiser, $\Q(\omega)$, for different values of $\omega$ in Fig. \ref{figure 2}. As $\omega$ increases, the LdG energy $E_{LdG}(\Q(\omega))$ increases to a constant value $E^*$, and the surface energy $E_{bc}(\Q(\omega))$ decreases to zero, which is consistent to Proposition \ref{pro:1}. For very weak anchoring ($\omega=10^{-4}$), the LdG energy of the global minimiser almost vanishes, that is, {$\Q(\omega)$ is almost constant and uniaxial everywhere. For modest anchoring ($\omega=1$), the global minimiser has non-zero LdG energy but the surface energy decreases compared to the previous case. For strong anchoring ($\omega=10^{4}$), the surface energy of the global minimiser almost vanishes, and defects appear at the cuboid vertices. Based on Fig. \ref{figure 2} and \cite{doi:10.1080/02678290903056095}, we take $W=0.01$ Jm$^{-2}$ so that $\omega=\frac{\lambda W}{\sqrt{2CL}}$ is approximately around $10^2$, which is in the strong anchoring regime.

\section{Numerical method}\label{sec:numerical method}

In this section, we describe the numerical methods used to compute the critical points of \eqref{eq: energy total}. The critical points, $\Q$, are solutions of the Euler-Lagrange equations (\ref{EL_b}), which are a system of five nonlinear partial differential equations, for the five components of the $\Q$-tensor, $Q_1,\cdots,Q_5$, in \eqref{eq:5-degree}.
 A critical point, $\Q$, is stable if the Hessian, $\nabla^2 E_\omega(\Q)$, only has positive eigenvalues (so that it is a local minimum), and unstable if $\nabla^2 E_\omega(\Q)$ has at least one negative eigenvalue. More precisely, a critical point $\Q$ is an index-$k$ saddle point for which  $\nabla^2 E_\omega(\Q)$ has exactly $k$ negative eigenvalues: $\lambda_1 \leqslant \cdots \leqslant \lambda_k<0$, corresponding to $k$ unit eigenvectors $\hat{\_v}_1,\cdots,\hat{\_v}_k$ subject to $\big\langle{\hat{\_v}_i}, \hat{{\_v}}_j \big\rangle = \delta_{ij}$, $1\leqslant i, j \leqslant k$. A stable critical point $\Q$ is an index-0 critical point, i.e. the smallest eigenvalue of $\nabla^2 E_\omega(\Q)$ is positive. Typically, a stable critical point can be easily found by the gradient descent method using a proper initial guess, however, it is not easy to provide good initial conditions since we do not have prior knowledge of the stable critical points.  In recent works \cite{shi2023hierarchies,han2022prism,shi2022nematic,han2020solution,jcm2023}, the solution landscape and saddle dynamics (SD) method \cite{2019High,yin2020constrained,zhang2021optimal,scm2023,JJIAM2023} have been successfully used to efficiently compute the critical points of LdG/rLdG free energy in 2D or 3D domains with Dirichlet, Neumann or mixed boundary conditions.

The SD for finding an index-$k$ saddle point $\hat{\Q}$, (denoted by $k$-SD) is defined to be,
\begin{equation}
  \left\{
  \begin{aligned}
  \dot{\Q}&=- (\I-2\sum_{i=1}^k {\_v}_i{\_v}_i^\top)\nabla E_\omega(\Q), \\
    \dot{\_v}_i&=-   (\I-{\_v}_i{\_v}_i^\top-\sum_{j=1}^{i-1}2{\_v}_j{\_v}_j^\top)\nabla^2 E_\omega(\Q) \_v_i,\ i=1,2,\cdots,k ,\\
  \end{aligned}
  \right.
\label{eq: SD}
\end{equation}
where $\I$ is the identity operator. To avoid evaluating the Hessian of $E_\omega(\Q)$, we use the dimer
\begin{equation}
  h(\Q,\_v_i)=\frac{\nabla E_\omega(\Q+l\_v_i)-\nabla E_\omega(\Q-l\_v_i)}{2l}
\end{equation}
as an approximation of $\nabla ^2 E_\omega(\Q)\_v_i$, with a small dimer length $2l$. By setting the $k$-dimensional subspace $\hat{\mathcal{V}}=\text{span} \big \{ \hat{\_v}_1,\cdots,\hat{\_v}_k \big \}$, $\hat{\Q}$ is a local maximum on $\hat{\Q}+\hat{\mathcal{V}}$ and a local minimum on $\hat{\Q}+\hat{\mathcal{V}}^\perp$, where $\hat{\mathcal{V}}^\perp$ is the orthogonal complement of $\hat{\mathcal{V}}$. The dynamics for $\Q$ in \eqref{eq: SD} can be written as
\begin{equation}
\begin{aligned}
\dot{\Q}&=\left(\I-\sum_{i=1}^k {\_v}_i{\_v}_i^\top\right) \left(-\nabla E_\omega(\Q)\right)+ \left(\sum_{i=1}^k {\_v}_i{\_v}_i^\top\right) \nabla E_\omega(\Q) \\
&= \left( \I-\mathcal{P}_{\mathcal{V}}  \right)\left(-\nabla E_\omega(\Q)\right)+ \mathcal{P}_{\mathcal{V}} \left(\nabla E_\omega(\Q)\right),
\end{aligned}
\end{equation}
where $\mathcal{P}_{\mathcal{V}}\nabla E_\omega(\Q)=\left(\sum_{i=1}^k {\_v}_i{\_v}_i^\top\right)\nabla E_\omega(\Q)$ is the orthogonal projection of $\nabla E_\omega(\Q)$ on $\mathcal{V}=\text{span} \big \{ \_v_1,\cdots,\_v_k \big \}$. Thus, $\left( \I-\mathcal{P}_{\mathcal{V}}  \right)\left(-\nabla E_\omega(\Q)\right)$ is a descent direction on $\mathcal{V}^\perp$, and $\mathcal{P}_{\mathcal{V}} \left(\nabla E_\omega(\Q)\right)$ is an ascent direction on $\mathcal{V}$. 

The dynamics for $\_v_i, i=1,2,\cdots,k$ in \eqref{eq: SD} can be obtained by minimising the $k$ Rayleigh quotients simultaneously with the gradient type dynamics,
\begin{equation}
\min_{{\_v}_i}\text{  }\left<\_v_i, \nabla ^2 E_\omega(\Q)\_v_i\right>,\ \text{s.t.}\ \left<\_v_i,\_v_j\right>=\delta_{ij}, \ j=1,2,\cdots,i,
\end{equation}
which generates the subspace $\mathcal{V}$ by computing the eigenvectors corresponding to the smallest $k$ eigenvalues of $\nabla^2 E_\omega(\Q)$.

In the remainder of this section, we outline our numerical discretization methods in detail. The non-dimensionalised domain $V=[-1,1]^2 \times [-h,h]$ is discretized into $N$ nodes with a small spatial distance $\delta x$. We elaborate on the numerical issues by using $k$-saddle dynamics to find a target saddle point, $\hat{\Q}$.

The condition number of $J(\hat{\Q})$, the Jacobian operator of the $k$-saddle dynamics, depends on the condition number of $\nabla_{\delta x}^2 E_{\omega}(\hat{\Q})$ (spatial discretization of $\nabla^2 E_{\omega}(\hat{\Q})$), i.e. $Cond_2(J(\hat{\Q}))\geq Cond_2(\nabla_{\delta x}^2 E_{\omega}(\hat{\Q}))$ \cite{shi2023hierarchies}, where
\begin{equation}
  \nabla_{\delta x} ^2 E_\omega(\hat{\Q})=\nabla_{\delta x} ^2 E_{LdG}(\hat{\Q})+\omega \nabla_{\delta x}^2 E_{bc}(\hat{\Q}).
\end{equation}
The spectral decomposition of $\nabla_{\delta x}^2 E_{\omega}(\hat{\Q})$ is 
\begin{equation}
    \nabla_{\delta x}^2 E_\omega(\hat{\Q})=\sum_{i=1}^{N} \lambda_i \_v_i^* {\_v_i^*}^\top, |\lambda_1|\leqslant |\lambda_2|\leqslant\cdots \leqslant |\lambda_N|,
\end{equation}
we have 
\begin{equation}
Cond_2(\nabla_{\delta x}^2 E_\omega(\hat{\Q})) = |\lambda_N|/|\lambda_1|.
\end{equation}
On one hand, with large $h$ and small $\omega$, the $\lambda_1$ for some target saddle point $\hat{\Q}$ is small since we can move the middle state on $z=0$ up and down without a significant energetic cost \cite{shi2023hierarchies}, i.e. $Cond_2(\nabla_{\delta x}^2 E_\omega(\hat{\Q}))$ is large.
On the other hand, as $\omega\to\infty$, the eigenvalues of the $\omega\nabla_{\delta x} ^2 E_{bc}(\hat{\Q})$ tend to infinity, while the eigenvalues of the $\nabla_{\delta x} ^2 E_{LdG}(\hat{\Q})$ remain bounded, i.e. $Cond_2(\nabla_{\delta x}^2 E_\omega(\hat{\Q}))$ is large for large enough $\omega$. Both cases can lead to ill-conditioned and stiff dynamics \cite{luo2022sinum}. 

To deal with the ill-conditioning, we adopt the same numerical scheme as in \cite{shi2023hierarchies} which studies the solution landscape of nematic liquid crystal in 3D cuboid with Dirichlet boundary conditions on the lateral surfaces and Neumann boundary conditions on the top and bottom, to deal with the $k$-saddle dynamics of $E_{LdG}(\Q)$. The linear term in \eqref{eq: SD} is implicitly discretized for numerical stability. The nonlinear term, $|\Q|^2\Q$, is also semi-implicitly discretized in time direction as $|\Q^n|^2\Q^{n+1}$, for better numerical stability. The term $|\Q^n|^2\Q^{n+1}$ is beneficial for solving linear equations in the semi-implicit scheme, because it is a positive definite term of the diagonal elements. Instead of re-generating unstable eigendirections with the gradient type dynamics in \eqref{eq: SD}, we apply a single-step Locally Optimal Block Preconditioned Conjugate Gradient (LOBPCG) method \cite{Knyazev2001TowardTO} to calculate the unstable eigendirections, and  Hessians are also approximated by dimers \cite{2019High}.

Next, we deal with the stable scheme for $E_{bc}(\Q)$.
In the right hand side of the dynamics in \eqref{eq: SD}, the gradient of $E_{bc}(\Q)$ is given by 
\begin{equation}
  \nabla E_{bc}(\Q)
  =\vec{\nu} \vec{\nu}^\top \Q+\Q\vec{\nu}\vec{\nu}^\top+\frac{2s_+}{3}\vec{\nu}\vec{\nu}^\top,
  \label{eq:nabla_E_bc}
\end{equation}
which has both coupling and non-coupling terms.
For example, on the top and bottom surfaces, $\vec{\nu}=(0,0,\pm1)$, the component of $\nabla E_{bc}(\mathbf Q)$ in the dynamics of $Q_1$ is 
$-2Q_{1}-2Q_{2}+\frac{2s_+}{3}$,
where $-2Q_{1}$ is the non-coupling term and $-2Q_{2}$ is the coupling term, i.e. a non-$Q_1$ variable, $Q_2$, appears in the dynamics of $Q_1$.
The non-coupling linear terms are implicitly discretized for numerical stability and the coupling terms are explicitly discretized, e.g., the component of $\nabla_{\delta x} E_{bc}(\Q_n,\Q_{n+1})$ in the dynamics of $Q_1$ is
$-2Q_{1,n+1}-2Q_{2,n}+\frac{2s_+}{3}$,
to decouple the (five) dynamics of $Q_i$, $i = 1,\cdots, 5$ and keep the block diagonal structure of the iteration matrix, so that the linear equations can be solved efficiently. 

Combining the above, the semi-implicit scheme is given by,
\begin{equation}
  \begin{cases}
    \begin{aligned}
        \frac{\Q_{n+1}-\Q_n}{\delta t}=& \Delta_{\delta x} \Q_{n+1}-\lambda^2\left(\frac{A}{2C}\Q_{n+1}+\frac{1}{2}|\Q_{n}|^2\Q_{n+1}-\frac{B}{2C}\left({\Q_{n}}^2-\frac{|\Q_{n}|^2}{3}\I \right)\right) \\
        &+\omega \nabla_{\delta x}E_{bc}(\Q_{n},\Q_{n+1})/\delta x+(2\sum_{i=1}^k {\_v}_{n,i}{\_v}_{n,i}^\top)\nabla_{\delta x} E_\omega(\Q_n),\\
        \text{Renew } \_v_{n,i} & \text{ as } \_v_{n+1,i} \text{ with single-step LOBPCG} , \  i=1,2,\cdots,k,\\
    \end{aligned}
  \end{cases}
    \label{eq: semi-implicit}
\end{equation}
where the time discretization $\delta t$ is treated using the Barzilai-Borwein step size for the ill-conditioning. 

The gradient of the surface energies does not break the symmetry of linear equations in semi-implicit scheme \eqref{eq: semi-implicit}, and hence, we can use the Minimal Residual Method (MINRES) to solve it efficiently. The SD pushes the iteration point into the basin of attraction of the target critical point and we use the Newton method to
complete tail convergence with a higher convergence rate. When the target critical point has small absolute eigenvalues, we will use the Inexact-Newton method, since the ill-conditioned linear equation is hard to solve exactly, with Newton iterations \cite{shi2023hierarchies,han2022prism}. In practice, we use the following hybrid scheme of semi-implicit scheme and Inexact-Newton scheme,
\begin{equation}
\begin{cases}
    \text{The semi-implicit scheme \eqref{eq: semi-implicit}},\ \Vert \nabla_{\delta x} E_\omega(\Q_n)\Vert \geqslant 10^{-3},\\
    \Q_{n+1} =\Q_{n} +\delta \Q, \Vert \nabla_{\delta x}^2 E_\omega(\Q_n) \delta \Q +\nabla_{\delta x} E_\omega(\Q_n) \Vert \leqslant \eta_n\Vert \nabla_{\delta x} E_\omega(\Q_n)\Vert \ ,\text{otherwise},\\
\end{cases}
\label{eq: Hybrid}
\end{equation}
where $\eta_n<1$ is the tolerance for solving linear equations in the Inexact-Newton method.

\section{Results}
The top and bottom faces of our domain $V$ are 2D squares. For $h=1$, the lateral faces are squares and for $h \neq 1$, the lateral faces are 2D rectangles with short or long edges along the $z$-axis. In what follows, we refer to the leading eigenvector of $\Q$ or the eigenvector with the largest positive eigenvalue, as the \emph{nematic director}. Recall that we work with large values of $\omega$, i.e. in the strong anchoring regime which enforces planar degenerate anchoring conditions so that the nematic director is tangent to $\partial V$. This means that the director is planar on a given face, but free to rotate in the plane of the face. On a given edge, the director is necessarily either parallel or anti-parallel to the edge, leading to the discontinuities at the vertices of $V$.

We briefly review the critical points of rLdG energy on a 2D square domain with tangential Dirichlet boundary conditions, consistent with the $\omega \to \infty$ limit of \eqref{eq: energy total} \cite{robinson2017molecular}.  
In the $h\to 0$ limit, i.e. the thin film limit, the rLdG framework has two degrees of freedom $\frac{Q_1-Q_2}{2}$ and $Q_3$ in \eqref{eq:5-degree} assuming $\frac{Q_1+Q_2}{2}$, $Q_4$ and $Q_5$ to be constants.
For $\lambda$ small enough, the Well Order Reconstruction Solution (WORS) is the unique (stable) 2D critical point of the rLdG free energy. The WORS is special since $\frac{Q_1-Q_2}{2}=0$ and $Q_3 = 0$ along the square diagonals for the WORS critical point, which implies that there is a diagonal defect cross connecting the four square vertices. The NLC molecules are disordered in the square plane, along the diagonals, and the cross partitions the square domain into quadrants such that the nematic director is constant in each square quadrant.
As $\lambda$ increases, the stable WORS becomes an index-$1$ saddle point and bifurcates into two stable diagonal states (i.e. there are two rotationally equivalent states, labelled as D$_1$ and D$_2$). The D states are approximately uniaxial, with the director along one of the square diagonals. At the second bifurcation point $\lambda = \lambda^{**}$, the index-1 WORS bifurcates into an index-2 WORS and the boundary distortion or bent director solutions (BD). The BD solutions are characterised by two defect lines (with $\frac{Q_1-Q_2}{2}=0$ and $Q_3=0$) near two opposite edges. Each index-1 BD further bifurcates into an index-2 BD and two index-1 Rotated (R) solutions, where each R solution is approximately uniaxial such that the uniaxial director rotates by $\pi$-radians between a pair of opposite square edges. As $\lambda$ further increases, the index-1 R critical points gain stability, and we have $4$ rotationally equivalent R solutions, R$_\text{n}$, R$_\text{s}$, R$_\text{w}$, and R$_\text{e}$ (subscript indicates the direction of the director bending in the square interior, north, south, west and east), related to each other by $\pi/2$-rotations. 

On a 2D rectangle with Dirichlet tangential boundary conditions, the rLdG free energy has a unique critical point when the short edge length $\lambda_r$ is small enough and the aspect ratio, $b$, is large enough. This unique and stable critical point has two defect lines localised along the short edges, the nematic director is primarily oriented along the long edges in the rectangular interior and is labelled as a BD state, consistent with the nomenclature for a square domain. As $\lambda_r$ increases, there is a critical value of $b^*(\lambda_r)>1$ such that the BD state is stable for $b > b^*$, and $b^*$ is an increasing function of $\lambda_r$ \cite{shi2022nematic}. When $\lambda_r$ is large enough or the aspect ratio, $b$, is small enough, the stable states are the $D$ and $R$ states, akin to a square domain.
If the profiles on the opposite surfaces of $V$ are the same, which is observed in the majority of the numerical results, we label the 3D critical point of \eqref{eq: energy total} on $V$ as A-B-C, where A, B and C are approximately 2D rLdG critical points on square and rectangular domains as discussed above. The label means that the critical point exhibits A, B, and C-like 2D profiles on the top and bottom surfaces, front and back surfaces, left and right surfaces, respectively. If the profiles on the opposite faces of $V$ are not the same, we label the corresponding 3D critical point of \eqref{eq: energy total} as A$_1$,A$_2$-B$_1$,B$_2$-C$_1$,C$_2$, where A$_1$, A$_2$, B$_1$, B$_2$, C$_1$, C$_2$ are approximately 2D rLdG critical points as discussed above, and the sequence labels the rLdG profiles on the top, bottom surfaces, front and back surfaces, left and right surfaces, respectively.

\subsection{Small h}\label{sec:smallh}

\begin{figure}[H]
  \centering
  \includegraphics[width=.7\textwidth]{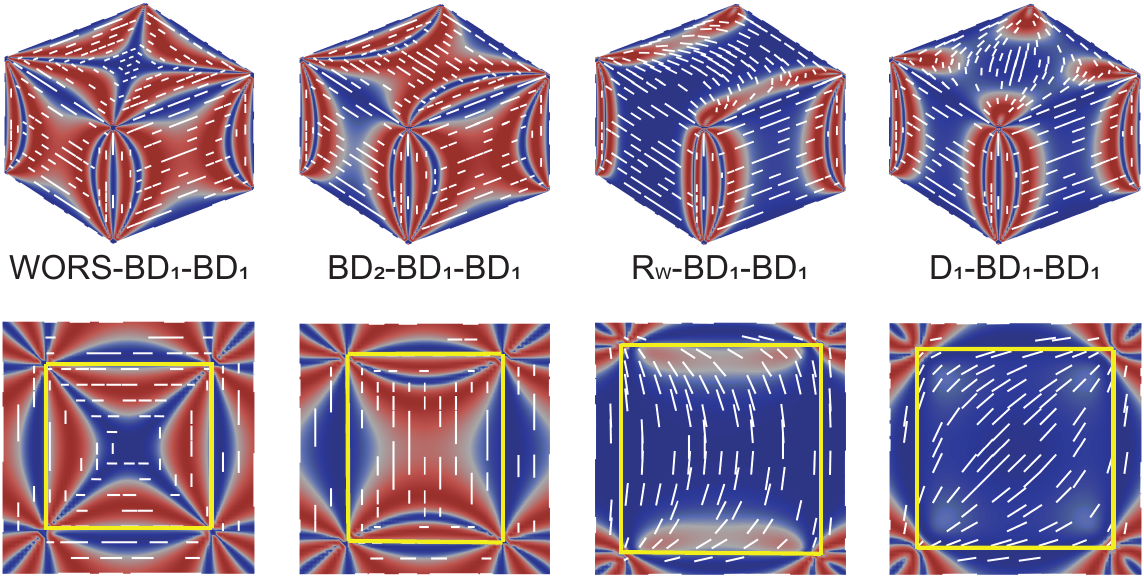}
  \caption{
  Top: Four typical BD$_1$-type solutions, WORS-BD$_1$-BD$_1$ for $\lambda^2 = 5$,
BD$_2$-BD$_1$-BD$_1$ for $\lambda^2 = 7$,
R$_\text{w}$-BD$_1$-BD$_1$ for $\lambda^2 = 32$,
D$_1$-BD$_1$-BD$_1$ for $\lambda^2 = 32$, and $h = 0.75$.
Bottom: The profiles on the middle cross-section of the cuboid. The yellow wireframe frames the 2D WORS-like, BD-like, R-like and D-like profiles in the center.
  }
  \label{figure 3}
\end{figure}
 

For $h<1$ small enough, we numerically observe BD$_1$-type profiles on the lateral faces of $V$, consistent with the fact that BD$_1$ is the rLdG energy minimiser on rectangular domains with tangent boundary conditions, when the short edge length is sufficiently small. The BD$_1$ profile has line defects (or bands of high biaxiality) localised near the short edges. As we vary $\lambda$, we recover WORS, BD, R, and D-type states on the top and bottom square faces.
Hence, for small enough $h$, we numerically compute the following family of critical points of \eqref{eq: energy total}:  WORS-BD$_1$-BD$_1$, BD$_2$-BD$_1$-BD$_1$, D$_1$-BD$_1$-BD$_1$, and R$_\text{w}$-BD$_1$-BD$_1$ on $V$ (Fig. \ref{figure 3}).
WORS-BD$_1$-BD$_1$ is stable for small $\lambda$ and small $h<1$,  D$_1$-BD$_1$-BD$_1$, and R$_\text{w}$-BD$_1$-BD$_1$ are stable for large $\lambda$ and small $h<1$ (Fig. \ref{figure 14}). 

The BD state is always an unstable critical point on the rLdG free energy on square domains and BD$_2$-BD$_1$-BD$_1$ state can be stable at $h<1$ and modest $\lambda$, raising interesting questions about the relationships between 2D rLdG critical points and 3D critical points of \eqref{eq: energy total}.
In the second row of Fig. \ref{figure 3}, on the 2D cross-section $z=0$, there are areas of high biaxiality near the four vertical edges (denoted by the four vertices of the 2D cross-section), and these high-biaxiality regions connect two line defects on the adjacent lateral surfaces. We draw yellow wireframes in each cross-section, the edges of which connect the high-biaxiality neighbourhoods of the vertices of the 2D cross-section. 
The domain inside the yellow wireframe is rarely affected by the high-biaxiality neighbourhoods of the four vertices. 
For $\lambda^2=7$, the yellow wireframe encloses a rectangular domain with relatively large aspect ratio and small short edge length, and a BD-type profile is observed within the yellow wireframe (as expected). 
This heuristic argument explains the stability (in the sense of positive second variation of \eqref{eq: energy total}) of the BD-BD$_1$-BD$_1$ critical point for certain values of $\lambda$ and $h$, as shown in the phase diagram Fig. \ref{figure 14}.

\begin{figure}
\includegraphics[width=.9\textwidth]{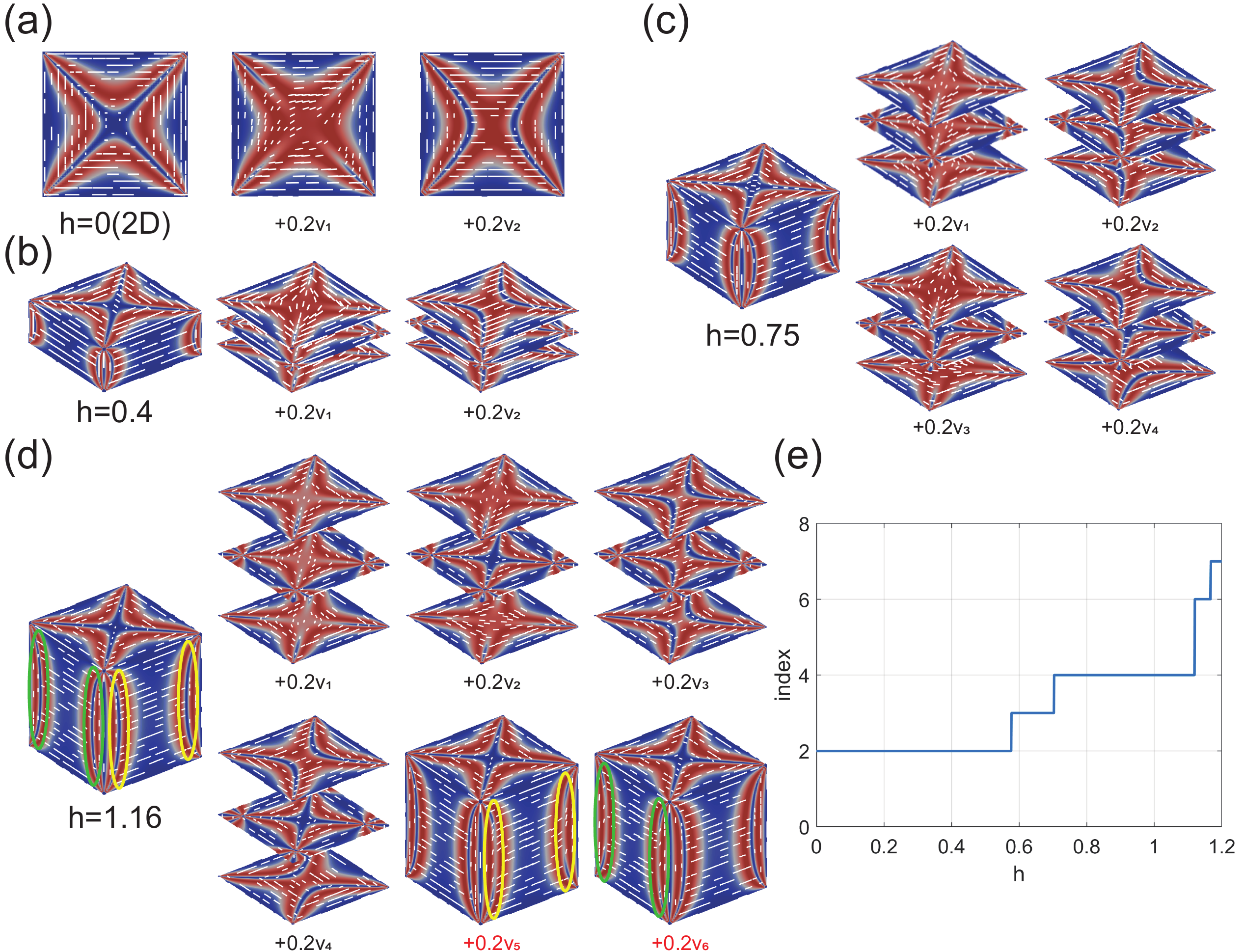}
  \caption{(a-d) The profiles of WORS-BD$_1$-BD$_1$ and the profiles with the perturbation along its unstable eigenvectors, i.e. WORS-BD$_1$-BD$_1 + 0.2\_v_i$, $i = 1,\cdots,6$ at $\lambda^2=16$ and different $h$. 
   In (d), we circle the profiles affected by $v_5$ and $v_6$ in yellow or green circles on the lateral surfaces.
  (e) The index of WORS-BD$_1$-BD$_1$ versus $h$ at $\lambda^2=16$.
  }
  \label{figure 5}
\end{figure}

As one can speculate from the rLdG study in \cite{han2020reduced, han2020solution}, the index of WORS-BD$_1$-BD$_1$ increases as $\lambda$ increases. Besides,  
we numerically observe that the index of WORS-BD$_1$-BD$_1$ increases as $h$ increases (Fig. \ref{figure 5}(e)).
  With $\lambda^2 = 16$, the 2D WORS is index-2 (Fig. \ref{figure 5}(a)).
 For $h=0.4$ (Fig. \ref{figure 5}(b)), in the framework of the full LdG model with the full five degrees of freedom, the  WORS-BD$_1$-BD$_1$ is also index-$2$ and the two unstable eigendirections are analogous to the two in-plane unstable eigendirections of the 2D index-2 WORS on a square domain. For $h=0.75$ (Fig. \ref{figure 5}(c)), the WORS-BD$_1$-BD$_1$ can accommodate $z$-variant unstable eigenvectors, v$_3$ and v$_4$, along which the WORS on the top surface relaxes to the D$_1$ state and BD$_1$ state respectively, and the bottom profile on $z=-1$ relaxes to the D$_2$ and BD$_2$ states respectively. For $h=1.16$, the BD$_1$-profile is energetically disadvantaged on the lateral faces in the $xz$ and $yz$-planes, for its two long line defects concentrated along the $z$-edges or $z$-axis, and the WORS-BD$_1$-BD$_1$ admits two further unstable eigenvectors, v$_5$ and v$_6$, which break the line defects, exploit the full five degrees of freedom and rotate the director out of the plane (Fig. \ref{figure 5}(d)). The property of increasing  Morse index with increasing $h$ need not hold for more exotic $z$-variant critical points of \eqref{eq: energy total} e.g. BD$_1$,BD$_2$-BD$_1$-BD$_1$ critical point (called BD-WORS-BD in \cite{shi2023hierarchies}), and in what follows, we use the WORS-BD$_1$-BD$_1$ critical point as the parent state for computing saddle points and stable critical points of \eqref{eq: energy total} for $h<1$.  
\subsection{Large h}
As reported in \cite{shi2022nematic}, when $h$ is large enough, the BD$_2$ state with line defects localised near the short edges, is energetically favorable on the lateral surfaces of $V$, for small enough $\lambda^2$. For $h>1$, the short edges are along the $x$ and $y$-edges on the lateral faces on $V$. We find BD$_2$-type profiles on the lateral surfaces of the numerically computed critical points of \eqref{eq: energy total} for $h>1$ e.g., the 3D critical points WORS-BD$_2$-BD$_2$, and D$_1$-BD$_2$-BD$_2$ in Fig. \ref{figure 7}.

Analogous to the results in Section \ref{sec:smallh} for $h<1$, we also find the following 3D critical points of \eqref{eq: energy total}: R-BD$_2$-BD$_2$ and BD-BD$_2$-BD$_2$. However, the BD-BD$_2$-BD$_2$ is always an unstable critical point. We conjecture that the instability of the BD-BD$_2$-BD$_2$  critical point arises from the fact that it is uniaxial near the center of the cuboid and exhibits a 2D BD-type profile near the top and bottom surfaces of $V$ which are squares, and a 2D BD state is always unstable on 2D square domains. The solution R-BD$_2$-BD$_2$ only exists for large enough $h$ and $\lambda$. 
The R-BD$_2$-BD$_2$ might be stable for very large values of $h$, but it is unstable in our calculated parameter domain. For rectangles with long short edges, the stable rLdG critical points are D and R states, and hence, as $\lambda$ increases for a fixed $h$, we numerically observe stable 3D critical points with D and R-type profiles on all six faces of $V$, as will be discussed in Section \ref{sec:h=1} and \ref{sec:uni}.
\begin{figure}
  \centering
\includegraphics[width=.4\textwidth]{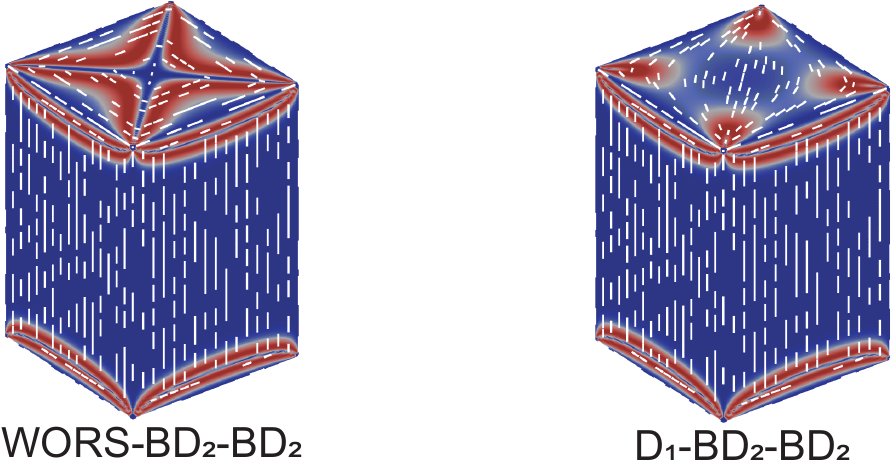}
  \caption{Two unstable BD$_2$-type solutions: WORS-BD$_2$-BD$_2$ and D$_1$-BD$_2$-BD$_2$ at $h=1.5, \lambda^2=70$. They are stable when $\lambda$ is small enough and $h$ is large enough.}
  \label{figure 7}
\end{figure}
\subsection{h = 1}\label{sec:h=1}


When $h = 1$, the domain $V$ is a cuboid with edges of equal lengths. For small $\lambda$, the unique stable critical point of \eqref{eq: energy total} is WORS-WORS-WORS.
The WORS-WORS-WORS is always a critical point and loses stability as $\lambda$ increases. For $\lambda^2 = 52$, the Morse index of WORS-WORS-WORS is greater than $12$. 
It has great symmetry and interesting 3D defect structures.
The line defects are on the two face diagonals of each face of the cuboid and the four body diagonals, which are surrounded by regions of high biaxiality in red in Fig. \ref{figure 10}. The area with high biaxiality decreases as $\lambda$ increases, since biaxiality is heavily penalised as $\lambda \to \infty$ \cite{majumdar2010landau}.
We construct the solution landscape with the WORS-WORS-WORS as the parent state  (Fig. \ref{figure 10}). Along the unstable directions, the cross structures of the WORS on the four lateral surfaces split into line defects and relax to WORS-BD$_1$(BD$_2$)-BD$_1$(BD$_2$).
The critical point, WORS-BD$_1$-BD$_1$, is connected to the critical points: WORS-BD$_1$-BD$_1$, BD$_2$-BD$_1$-BD$_1$, D$_2$-BD$_1$-BD$_1$, R$_\text{w}$-BD$_1$-BD$_1$, analogous to the solution landscapes with $h<1$.
Since $h=1$, BD$_1$ and BD$_2$ are energetically degenerate on the lateral faces, and we find corresponding BD$_2$-type solutions BD$_2$-BD$_2$-BD$_2$, D$_2$-BD$_2$-BD$_2$, R$_\text{w}$-BD$_2$-BD$_2$, connected with the critical point, WORS-BD$_2$-BD$_2$. Starting from the critical points: R$_\text{w}$-BD$_{1(2)}$-BD$_{1(2)}$ and D$_2$-BD$_{1(2)}$-BD$_{1(2)}$, and following either gradient flow or saddle dynamics, we find almost uniaxial solutions with 2D D-like or R-like profiles on all the six faces of the cube. 

\begin{figure}[H]
  \centering
  \includegraphics[width=.95\textwidth]{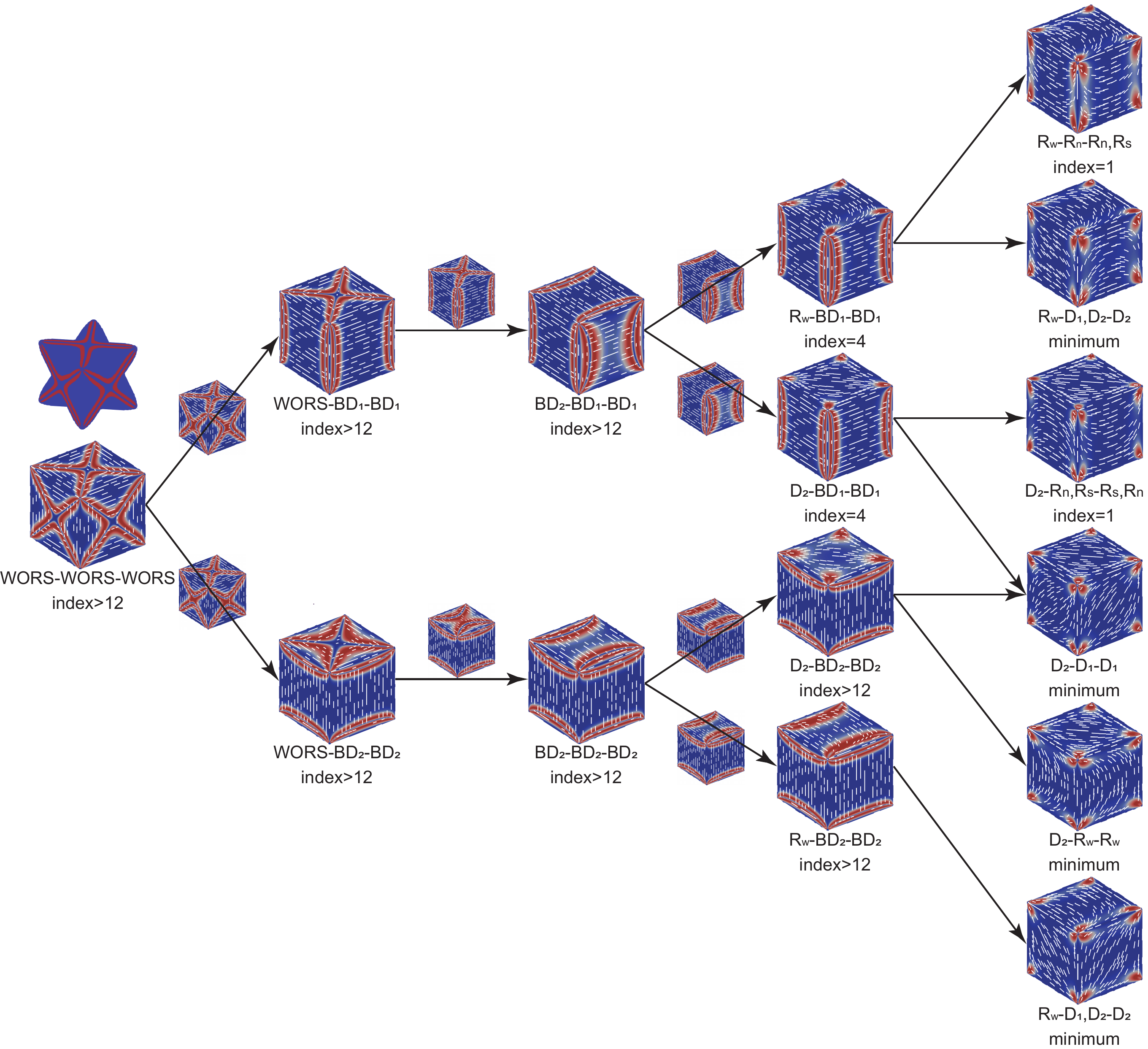}
  \caption{The solution landscape at $h=1$ and $\lambda^2=70$. The arrow represents that following the saddle dynamic in Eq. \eqref{eq: SD}, the higher-index critical state at the end of the arrow with a small perturbation along an unstable direction 
converges to the lower-index critical state at the arrowhead. 
The defect profile of WORS-WORS-WORS is shown by the red high-biaxiality regions $(\beta^2\geqslant 0.8)$ and the surrounding blue low-biaxiality regions.
}
  \label{figure 10}
\end{figure}

The next questions pertain to the construction of initial conditions that mimic these stable, almost uniaxial critical points of \eqref{eq: energy total}, using topological arguments, and estimating the multiplicity of the almost uniaxial stable critical points of \eqref{eq: energy total}. 
A uniaxial $\Q$-tensor is described by $\Q = s_+ ( \mathbf{n}\otimes \mathbf{n} - \frac{\mathbf{I}}{3})$, where $\mathbf{n} \in \mathbb{S}^2$ is the eigenvector of $\Q$ with the non-degenerate eigenvalue, referred to as \emph{nematic director}. In the $\omega \to \infty$ limit, we have planar degenerate conditions on all six faces of the cuboid which require $\mathbf{n}$ to be tangent to all faces of the cuboid. 
In \cite{robbins2004classification}, the authors provide a complete topological classification of tangent nematic directors on cuboids or three-dimensional geometries in terms of a complete set of topological invariants: the edge signs, kink numbers and trapped areas. Let $E_j$ be an edge of the cube oriented in the $j$-direction, where $j$ could be $x, y, z$, i.e. one of the coordinate directions. The edge sign $e_{\mathbf{j}}^{E_j}= \pm 1$, determines the sign of $\mathbf{n}$ on the edge $E_j$, relative to the coordinate unit vector $\mathbf{j}$. The integer-valued kink number, $k_j^{\mathbf{v}}$, is a measure of the rotation of $\mathbf{n}$ along a path, that connects two edges meeting at the vertex $\mathbf{v}$, on the face normal to $\mathbf{j}$. The minimum possible winding e.g. a $\frac{\pi}{2}$ rotation between a pair of adjacent edges, is assigned zero kink number. 
The kink numbers satisfy a sum rule on each face, stemming from regularity assumptions on $\mathbf{n}$ \cite{robbins2004classification}. The D and R critical points of the rLdG free energy have zero kink numbers or minimum allowed rotation between pairs of adjacent square edges.
Let $S_v$ be a surface that isolates the vertex $\mathbf{v}$ of the cube, from the remaining vertices. Then, the trapped area, denoted by $\Omega^\mathbf{v}$, is the oriented area of the image, $\mathbf{n}(S_v)$, on the unit sphere $\mathbb{S}^2$. 
For a cuboid, the trapped areas are necessarily odd multiples of $\pi/2$. For the simplest topologies, the trapped area can only be $-\pi/2$ or $+\pi/2$. In what follows, we identify the different families of $\mathbf{n}$ that satisfy the tangent boundary conditions on the cuboid faces, with the simplest topologies, i.e. zero kink numbers and minimal trapped areas of $\pm \frac{\pi}{2}$. Once we identify the relevant $\mathbf{n}$ with the simplest topologies, the uniaxial initial conditions, $\Q_u = s_+(\mathbf{n}\otimes \mathbf{n} - \frac{\mathbf{I}}{3})$, can be constructed for the numerical solver.

Since we are interested in tangent nematic directors with the simplest topology and with zero kink numbers on the cuboid faces, we restrict ourselves to $\mathbf{n}$ which have a D or R-type profile on each cuboid face. The profile near each vertex has four choices in Fig. \ref{figure 11}(a) (without the profiles related by rotation). Two of them are called 3D splay vertices on the left, and the other two are called 3D non-splay vertices on the right of Fig. \ref{figure 11}(a). We name the splay vertex with directors, $\mathbf{n}$, pointing from (towards) the vertex as ``source" (``sink''), marked by a red (black) circle respectively.

\begin{figure}
\centering
\includegraphics[width=.99\textwidth]{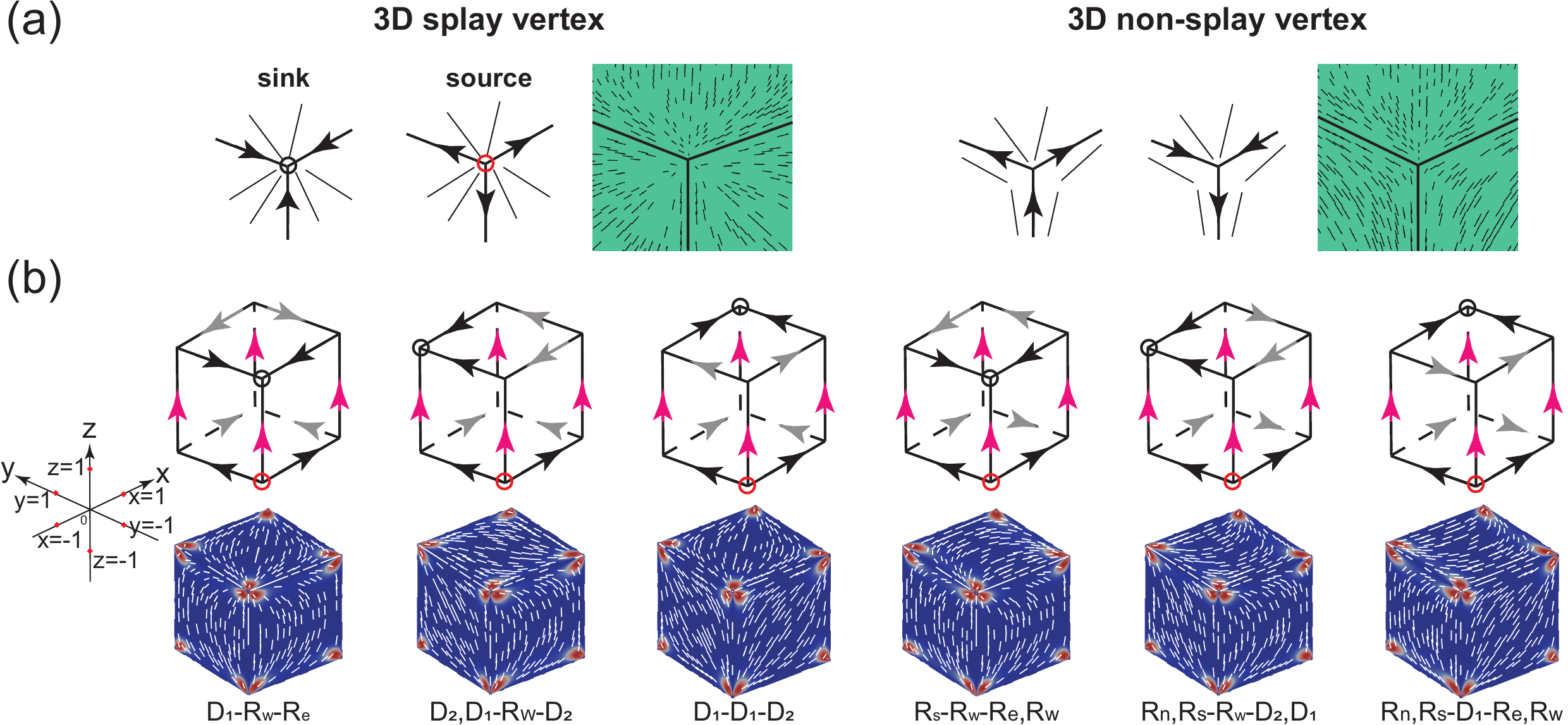}
  \caption{(a) The classification of vertices: splay and non-splay vertices. The black and red circle represents the "sink" and "source" vertex respectively.
  (b) The profiles of uniaxial stable states at $h=1$ and $\lambda^2=100$. The arrows on the edges of the cube frame represent the director $\mathbf{n}$.  We assume the directions on the four pink edges are $\vec{e}_z$, for ruling out some candidates with high energy.}
  \label{figure 11}
\end{figure}

Next, we enumerate the possibilities for tangent $\mathbf{n}$ with a D or R-type profile on each cuboid face, such that the vertices are either splay vertices or non-splay vertices. To reduce the number of candidates, we assume that $\mathbf{n} = \hat{\mathbf{z}}$, where $\hat{\mathbf{z}}$ is the unit vector in the $z$-direction, on the four vertical edges. 
This assumption is reasonable in the sense that it is consistent with minimal distortion across the cuboid height, and we are interested in minimal energy configurations. 
With these assumptions, there is one and only one ``sink" vertex on the top surface, and one and only one ``source" vertex on the bottom surface. Otherwise, the 2D profile on the top or bottom cuboid faces is neither the D nor R states.
If we fix the ``source" vertex on the bottom surface at $(-1,-1,-1)$, then there are three choices for the location of the ``sink" vertex on the top surface depending on the relative location between ``source" and ``sink": $(-1,-1,1)$, $(-1,1,1)$, $(1,1,1)$.
Once we fix the location of the sink vertex on the top vertex, the edge orientations are fixed on all the vertical edges and the edges intersecting at the source and sink vertex. 
There is freedom for the two gray edges on the bottom face, far from the ``source" vertex, but the edge orientations on the gray edges need to be chosen to ensure that $\mathbf{n}$ is consistent with either a D or R profile on the bottom surface. Namely, there are three choices for the edge orientations on the two gray edges and once we choose one orientation for $\mathbf{n}$, the other orientation follows from the requirement of having a D or R-type profile on the bottom surface, e.g. $(\mathbf{n}\vert_{(x,1,-1)},\mathbf{n}\vert_{(1,y,-1)})$: ($\hat{\mathbf{x}}$, $\hat{\mathbf{y}}$), ($\hat{\mathbf{x}}$,-$\hat{\mathbf{y}}$),($-\hat{\mathbf{x}}$,$\hat{\mathbf{y}}$). Similarly, there are three choices for the edge orientations of the two edges that do not meet at the ``sink" vertex on the top surface. 
In conclusion, there are $3\times 3\times 3 = 27$ candidates for the stable uniaxial critical points of \eqref{eq: energy total}, constructed from tangent nematic directors with the simplest topology.

Next, we outline the construction of the associated initial conditions for the $\Q$-solver. Let $\mathbf{Q}_{2D}$ be a minimiser of the LdG energy on a 2D square domain with tangent boundary conditions, for relatively large $\lambda$ (see \cite{han2020reduced}). There are six choices of $\mathbf{Q}_{2D}$ - the two D and four R stable critical points of the rLdG energy, of the form
\[
\mathbf{Q}_{2D} = s_+ (\mathbf{n}\otimes \mathbf{n} - \frac{\mathbf{I}}{3})
\]
where $\mathbf{n}$ is the 2D director that describes either the D or R solutions. We define the following Dirichlet boundary condition on a 3D cuboid by 
\begin{align}\label{eq:bc}
Q_{bc} = TQ_{2D}T^T,
T = 
\begin{cases}
\mathbf{I},  & \text{$\partial V_{tb}$} \\
\begin{pmatrix}
1 & 0 & 0\\
0 & 0 & 1\\
0 & 1 & 0
\end{pmatrix}, & \text{$\partial V_{fb}$}\\
\begin{pmatrix}
0 & 0 & 1\\
0 & 1 & 0\\
1 & 0 & 0
\end{pmatrix}, & \text{$\partial V_{lr}$}
\end{cases}
\end{align} 
where $T$ is the rotation matrix.
We can minimise the LdG energy, \eqref{eq:LdG}, on a 3D cuboid, with these fixed Dirichlet conditions and use the minimisers as initial conditions, to search for stable and almost uniaxial critical points of \eqref{eq: energy total}, for sufficiently large $\lambda$ and $h$. 

By using the aforementioned $27$ initial conditions, 
we only find $6$ stable almost uniaxial critical points of \eqref{eq: energy total}. 
When the ``sink" and ``source" vertex are on the same edge, the corresponding critical point of \eqref{eq: energy total} is D$_1$-R$_\text{w}$-R$_\text{e}$ and R$_\text{s}$-R$_\text{w}$-R$_\text{e}$,R$_\text{w}$;
when the ``sink" and ``source" vertices are on the face diagonal, the corresponding numerically computed critical point of \eqref{eq: energy total} is D$_2$,D$_1$-R$_\text{w}$-D$_2$ and R$_\text{n}$,R$_\text{s}$-R$_\text{w}$-D$_2$,D$_1$; and when the sink and source vertices are on the body diagonal of the cuboid, the corresponding critical point of \eqref{eq: energy total} is D$_1$-D$_1$-D$_2$ and R$_\text{n}$,R$_\text{s}$-D$_1$-R$_\text{e}$,R$_\text{w}$.
The critical points exhibiting R-type profiles on cuboid faces have higher energy than the critical points which exhibit D-type profiles on the cuboid faces, consistent with the fact that D-states have lower energy than R-states on 2D square domains with tangent boundary conditions \cite{lewis2014colloidal}. 
Thus, the states in the first row of Fig. \ref{figure 11}(b) have lower energy than the states in the second row, and D$_1$-D$_1$-D$_2$ has the lowest energy. We revisit these almost uniaxial critical points of \eqref{eq: energy total} and their stability in Sec. \ref{sec:phase_diagram}.

\begin{figure}[H]
  \centering
  \includegraphics[width=.99\textwidth]{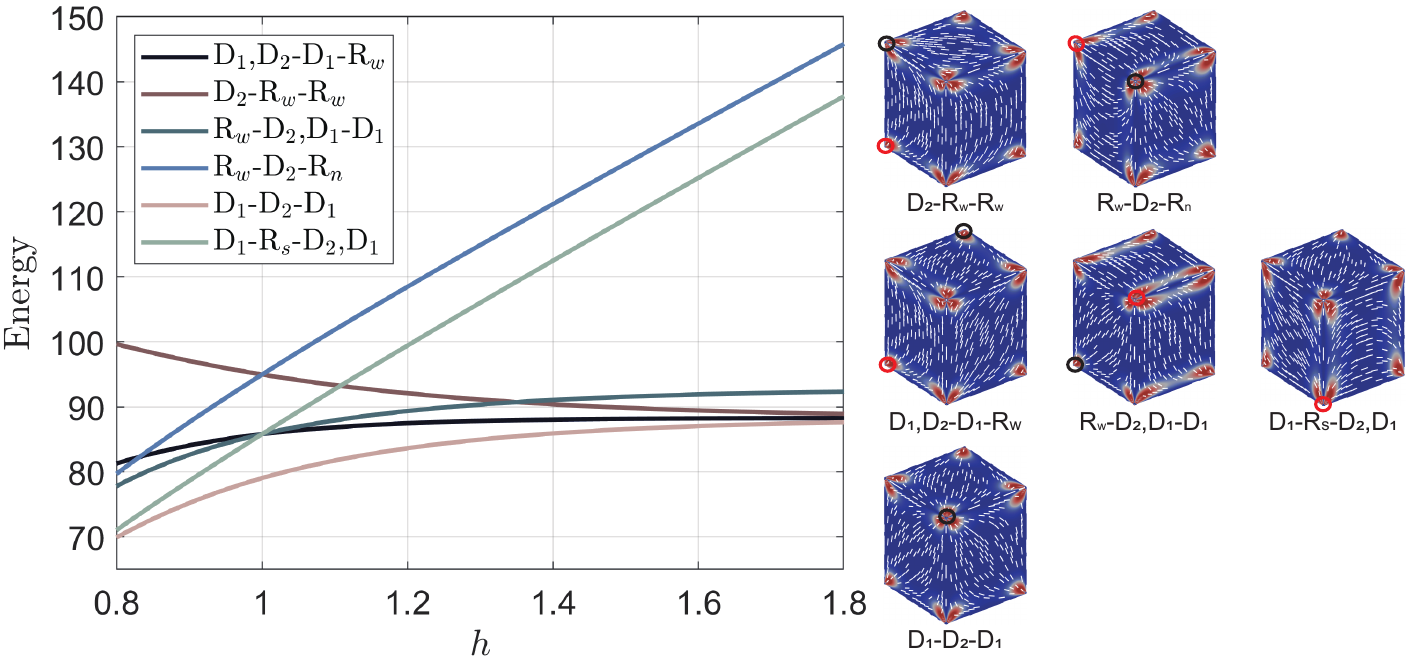}
  \caption{Left: The energy plot of the states on the right versus $h$ at $\lambda^2=60$. Right: Six uniaxial stable  states with $h = 1.2$, which are rotationally equivalent to the states in the first row of Fig. \ref{figure 11} with $h = 1$.}
  \label{figure 12}
\end{figure}
\subsection{The effect of the height on uniaxial states and transition pathways}\label{sec:uni}

When $h\neq 1$, the height of the cuboid and the edge length of the top or bottom surface square are not equal, and we have three more uniaxial states in Fig. \ref{figure 12} which are derived from the three states in the first row of Fig. \ref{figure 11}. With multiple stable almost uniaxial critical points of \eqref{eq: energy total}, we study the transition pathways between some of them in Fig. \ref{figure 12}. 
 In the following, we mark the ``sink" and ``source" vertices with yellow circles in Fig. \ref{figure 13}, and refer to both vertices as ``splay vertices". 

\begin{figure}
\includegraphics[width=.99\textwidth]{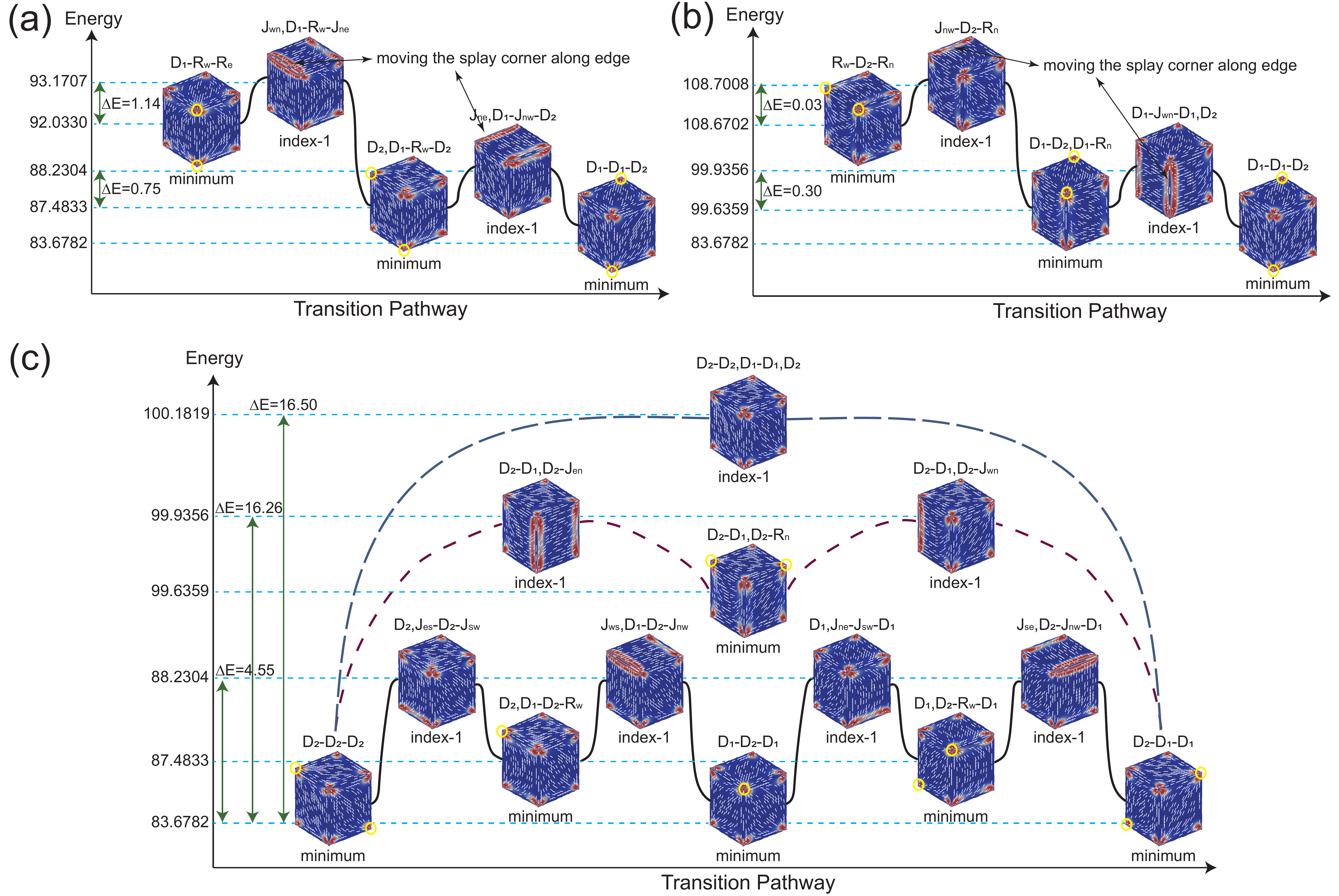}
  \caption{The transition pathways (a) between stable states D$_1$-R$_\text{w}$-R$_\text{e}$ and  D$_1$-D$_1$-D$_2$, and (b) between R$_\text{w}$-D$_2$-R$_\text{n}$ and D$_1$-D$_1$-D$_2$ at $\lambda^2=60,h=1.2$. (c) Three transition pathways between dual stable states D$_2$-D$_2$-D$_2$ and D$_2$-D$_1$-D$_1$ at $\lambda^2=60,h=1.2$.
  The yellow circles represent 3D splay vertices.
  }
  \label{figure 13}
\end{figure}
 
 In Fig. \ref{figure 13}(a), we show the transition pathway between locally stable state D$_1$-R$_\text{w}$-R$_\text{e}$ and globally stable state D$_1$-D$_1$-D$_2$.
The location of the splay vertex on the bottom surface does not change along the transition pathway, and the splay vertex on the top surface moves along the short edges from $(-1,-1,h)$ via $(-1,1,h)$, i.e. hits a locally stable critical point D$_2$,D$_1$-R$_\text{w}$-D$_2$, and then the splay vertex on the top surface moves to $(1,1,h)$, to settle into the stable critical point D$_1$-D$_1$-D$_2$.

The transition pathway between R$_\text{w}$-D$_2$-R$_\text{n}$ and D$_1$-D$_1$-D$_2$ in Fig. \ref{figure 13}(b) is composed of two segments: a transition pathway between R$_\text{w}$-D$_2$-R$_\text{n}$ and D$_1$-D$_2$,D$_1$-R$_\text{n}$, and a transition pathway between  D$_1$-D$_2$,D$_1$-R$_\text{n}$ and D$_1$-D$_1$-D$_2$. On the first segment, the splay vertex on the top surface moves along the short edge from $(-1,1,h)$ to $(1,1,h)$, and the other splay vertex doesn't move. On the second segment, the splay vertex on the top surface moves along the long edge from $(-1,-1,h)$ to $(-1,-1,-h)$. These numerical results suggest that the splay vertex moves one step (either along a short or long edge), on every step of the transition pathway.

In Fig. \ref{figure 13}(c), we investigate the switching mechanism between two different D-D-D states, D$_2$-D$_2$-D$_2$ and D$_2$-D$_1$-D$_1$. Analogous to the one-step-at-a-time pattern in Fig. \ref{figure 13}(a) and (b), there is a transition pathway plotted in black line with splay vertices moving along short edges, via three energy minima and four index-1 transition states, and a transition pathway in red with splay vertices moving long edges, via one energy minimum and two index-1 transition states. The switching could be impeded by getting trapped in an intermediate energy minimum. There is another pathway in blue, for which the two splay vertices move along the diagonals on the top and bottom simultaneously, via no intermediate energy minimum and only one transition state, but this pathway has a higher energy barrier. The multiple choices for transition pathways between two stable states (or critical points of the rLdG energy) have been reported in \cite{han2020solution} on a 2D hexagon with tangent boundary conditions. In \cite{han2020solution}, the authors also find a direct pathway between two stable critical points of the rLdG energy, connected via an index-$8$ saddle point.  In the 3D pathway, see the blue line in Fig. \ref{figure 13}, the transition state D-D$_2$,D$_1$-D$_1$,D$_2$ is an index-$1$ saddle point of \eqref{eq: energy total}, and hence, could be of relevance for practical processes.
 



\subsection{Phase diagram}
\label{sec:phase_diagram}

To summarise our numerical results, we compute a phase diagram in Fig. \ref{figure 14} as a function of $\lambda^2$ and $h$, where we demarcate stable and metastable states. 
In what follows, we label a critical point as being \emph{metastable} if it is an index-$0$ critical point and a critical point as being \emph{stable} if it is the minimum energy index-$0$ critical point amongst the catalogue of numerically computed metastable critical points.

For $h<1$, the unique stable state is WORS-BD$_1$-BD$_1$ for $\lambda$ small enough. As $\lambda$ increases, D$_1$-BD$_1$-BD$_1$ is the stable state. For $\lambda$ sufficiently large, R$_\text{w}$-BD$_1$-BD$_1$ is a metastable state with higher energy than D$_1$-BD$_1$-BD$_1$. There is a small area, for $h$ close to unity and modest $\lambda$, for which when D$_1$-BD$_1$-BD$_1$ loses global stability, i.e. it is metastable, or loses stability and BD$_2$-BD$_1$-BD$_1$ is the stable state, for which we offer heuristic insights in Section \ref{sec:smallh}.

When $h = 1$, the stable state is WORS-WORS-WORS, for $\lambda$ small enough. For $\lambda$ large enough, we find multiple uniaxial (meta)stable states like D$_1$-D$_2$-D$_1$, all of which have D and R-type profiles on the six cube faces. The almost uniaxial (meta)stable critical points exist for large $\lambda$, and are likely to be observable in experiments and applications based on large cuboids, with weak tangent boundary conditions. 

For $h>1$, BD$_2$ is the energetically preferred 2D critical point of the rLdG energy on the lateral faces, for small $\lambda$. For small $\lambda$, the stable state is WORS-BD$_2$-BD$_2$. For large values of $\lambda$ and $h>1$ (within our numerically computed range), the stable state is one of the uniaxial critical points with D-type profiles on the cuboid faces. This could change in the $h \to \infty$ limit. Additionally, there is a new stable state  BD$_1$-D$_1$-BD$_2$, for modest $\lambda^2$, and a metastable state D$_1$-BD$_2$-BD$_2$, for certain values of $\lambda$ and $h$.

\begin{figure}[H]
  \centering
  \includegraphics[width=.99\textwidth]{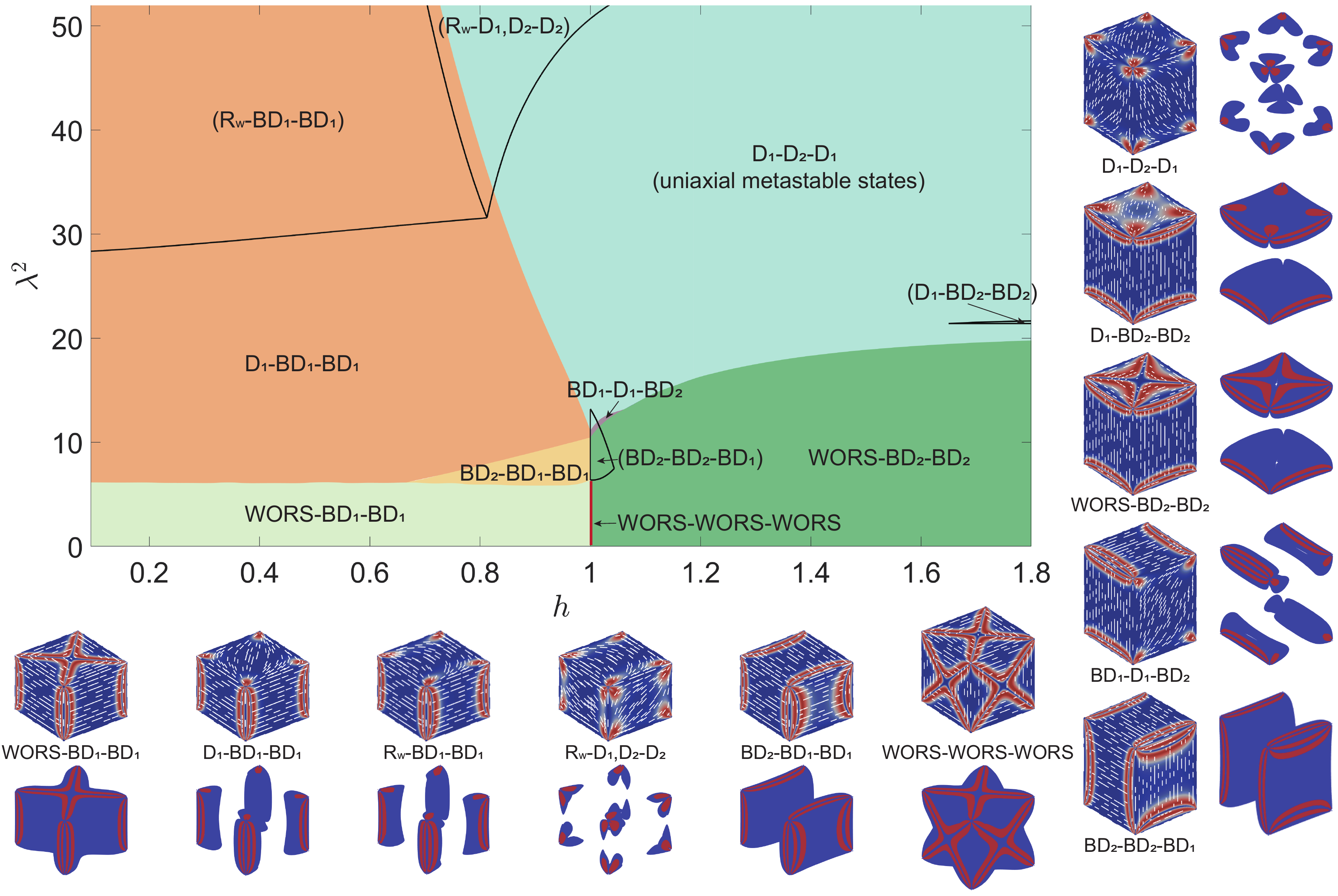}
  \caption{The phase diagram as a function of $\lambda^2$ and $h$. The parameter domains of different stable states are distinguished by different colours. The different metastable states are framed by black lines, with names in brackets. We present the numerically computed profiles and defect structures of the critical points of \eqref{eq: energy total} at $\lambda^2=50$, $h=0.78$ below the phase diagram; the same for $\lambda^2=50$, $h=1.25$ on the right-hand side of the phase diagram, and those for WORS-WORS-WORS at $\lambda^2 = 50$, $h = 1$.}
  \label{figure 14}
\end{figure}

\section{Conclusion and discussion}
\label{sec:conclusions}
We study NLC configurations inside a 3D cuboid, within the full LdG framework for which the LdG order parameter has five degrees of freedom, with planar degenerate/tangential boundary conditions. The tangent boundary conditions are enforced by means of surface energies, by means of a large surface anchoring coefficient fixed to be $W=0.01 \text{ Jm}^{-2}$ throughout this paper. There are two geometrical parameters in our study - the re-scaled edge length of the square cross-section, $\lambda$, and $h$ - the ratio of the cuboid cross-section edge length and height. 
We prove a batch of analytic results for a smoothed cuboid - the existence of a global minimiser of \eqref{eq: energy total} for $\omega>0$, uniqueness of the critical points of \eqref{eq: energy total} for $\lambda$ small enough, and importantly, that energy minimisers satisfy tangent boundary conditions in the $\omega \to \infty$ limit. We work with a fixed large value of $\omega$ throughout the manuscript, that ensures that the nematic director is tangent to the cuboid faces. 
This, in turn, leads to numerical difficulties, and we design a new numerical scheme to deal with the stiffness of the problem and accelerate the convergence rate of the saddle dynamics.

In Fig. \ref{figure 14}, we plot a phase diagram of the (meta)stable critical points of \eqref{eq: energy total} in the $h - \lambda^2$ plane. Fixing $h$, one can identify the profiles on the top and bottom cuboid faces with the solution landscape on a square domain, as a function of $\lambda^2$. Indeed, for small $\lambda$, we get WORS-type profiles on the top and bottom, and for large $\lambda$, we primarily get D and R-type profiles on the top and bottom cuboid surfaces. Similarly, if we fix $\lambda$ and traverse the phase diagram in the $h$-direction, the profiles on the lateral surfaces follow the predictions for solution landscapes on rectangles in \cite{shi2022nematic}. For small $\lambda$, the profiles on the lateral faces are either BD$_1$ or BD$_2$, depending on $h$ ($h$ determines whether the $z$-edge is the shorter edge or not, and the BD line defects are localised along the shorter edges). If $h=1$, then we observe the WORS on the lateral faces. For $\lambda$ sufficiently large, the profiles on the lateral faces of the (meta)stable critical points of \eqref{eq: energy total} depend on $h$; we get BD-profiles on the lateral faces for $h$ small, and the D and R-profiles for $h$ sufficiently large. Of course, if $h$ is sufficiently large or sufficiently small, the D and R-profiles on the lateral faces closely resemble BD-type profiles with an approximately constant director along the longer edge of the lateral surface. The phase diagram in Figure~\ref{figure 14} illustrates how solution landscapes on squares and rectangles control the profiles on the cuboid faces, for (meta)stable critical points of \eqref{eq: energy total} with tangent boundary conditions. This, in turn, determines the interior 3D structure, including defect structures, of physically relevant NLC configurations within 3D cuboids.

There are numerous interesting future research directions. We could work with weaker anchoring i.e. $\omega < 1$, which would offer greater freedom on the cuboid faces. In particular, we do not expect close correspondence between 2D solution landscapes and (meta)stable critical points of \eqref{eq: energy total}, for smaller values of $\omega$. A further generalization concerns arbitrary 3D geometries with polygonal faces and tangent boundary conditions. The analysis in \cite{han2020reduced} for arbitrary 2D regular polygons, in the rLdG setting, can be applied to (meta)stable critical points of \eqref{eq: energy total} on arbitrary 3D geometries with polygonal faces, particularly when the edge lengths are small or very large. A further thought concerns the applications of machine learning to train data for generating solution landscapes of complex systems. We have a series of papers on NLC solution landscapes in 2D and 3D \cite{shi2023hierarchies,han2022prism}, and our results on stable states, unstable saddle points and pathways between critical points can be used as precious training data for new machine learning-based algorithms.

\section{Acknowledgements}
B. Shi thanks the University of Strathclyde for their hospitality. The authors would also like to thank the Isaac Newton Institute for Mathematical Sciences for their hospitality during the programme ``Uncertainty Quantification and Stochastic Modelling of Materials" when work on this paper was undertaken.

\bibliographystyle{unsrt}
\bibliography{references}

\end{document}